\newif\ifproceedingsversion
\proceedingsversionfalse
\newif\ifcombinedversion
\combinedversiontrue
\newif\ifcolrefversion
\colrefversionfalse
\newif\ifcompactversion
\compactversionfalse
\ifcompactversion
\documentclass[envcountsame,runningheads]{llncs}
\else
\documentclass[11pt,envcountsame,runningheads]{llncs}
\usepackage{a4}
\fi
\usepackage{graphics}
\usepackage{amsmath,amssymb}
\usepackage{url}
\usepackage[utf8]{inputenc}

\usepackage{xspace}
\usepackage{tikz}
\usepackage{float}
\usepackage{enumitem}
\spnewtheorem{myclaim}[theorem]{Claim}{\mdseries}{\itshape}
\spnewtheorem{myexample}[theorem]{Example}{\bfseries}{\itshape}
\usepackage{thmtools}

\usepackage{etoolbox}
\newcommand{\ifproceedings}[2]{\ifbool{proceedingsversion}{#1}{#2}}
\newcommand{\ifcombined}[2]{\ifbool{combinedversion}{#1}{#2}}
\newcommand{\ifcolref}[2]{\ifbool{colrefversion}{#1}{#2}}
\newcommand{\ifcompact}[2]{\ifbool{compactversion}{#1}{#2}}

\newenvironment{proofof}[1]{\renewcommand{\proofname}{Proof of #1}\begin{proof}}
{\end{proof}}
\newcounter{oq}

\newcommand{\reals}{\mathbb{R}}

\newcommand{\function}[2]{:#1 \rightarrow #2}

\newcommand{\setdef}[2]{\left\{ \hspace{0.5mm} #1 : \hspace{0.5mm} #2 \right\}}
\newcommand{\Set}[1]{\left\{ \hspace{0.5mm} #1 \hspace{0.5mm} \right\}}
\newcommand{\refeq}[1]{(\ref{eq:#1})}

\newcommand{\calY}{\mathcal Y}

\newcommand{\msetdef}[2]{\left\{\!\!\left\{ \hspace{0.5mm} #1 : \hspace{0.5mm} #2 \right\}\!\!\right\}}
\newcommand{\cclass}[1]{\textsf{\upshape #1}}

\newcommand{\p}{\cclass{P}\xspace}
\newcommand{\conp}{\cclass{coNP}\xspace}
\newcommand{\gi}{\cclass{GI}\xspace}
\newcommand{\mathmyfont}{\mathsf}

\newcommand{\ds}[1]{S(#1)}
\newcommand{\auto}[1]{\mathit{Aut}(#1)}
\newcommand{\amenable}{\ensuremath{\mathmyfont{Amenable}}\xspace}
\newcommand{\compact}{\ensuremath{\mathmyfont{Compact}}\xspace}
\newcommand{\godsil}{\ensuremath{\mathmyfont{Godsil}}\xspace}
\newcommand{\tinhofer}{\ensuremath{\mathmyfont{Tinhofer}}\xspace}
\newcommand{\refinable}{\ensuremath{\mathmyfont{Refinable}}\xspace}
\newcommand{\discrete}{\ensuremath{\mathmyfont{Discrete}}\xspace}
\newcommand{\Pa}{{\mathcal P}}
\newcommand\bigzero{\text{\large0}}
\DeclareMathOperator{\IMP}{Imp}
\DeclareMathOperator{\CFI}{CFI}
\DeclareMathOperator{\Aut}{Aut}
\DeclareMathOperator{\Ext}{Ext}
\DeclareMathOperator{\Iso}{Iso}

\newcommand{\extr}[1]{\Ext(#1)}
\newcommand{\diag}[3]{#1\oplus #2\oplus #3}
\newcommand{\ddiag}[2]{#1\oplus #2}
\newcommand{\dddiagggg}[7]{#1\oplus #2\oplus #3\oplus #4\oplus #5\oplus #6\oplus #7}
\newcommand{\ddiagg}[4]{#1\oplus #2\oplus #3\oplus #4}

\newcommand{\hide}[1]{}

\title{\ifcombined{Graph Isomorphism, Color Refinement, and Compactness}{}
\ifcompact{On Tinhofer's Linear Programming Approach to Isomorphism Testing}{}%
\ifcolref{On the Power of Color Refinement}{}%
  \thanks{This work was supported by the Alexander von Humboldt
    Foundation in its research group linkage program. The second
    author and the fourth author were supported by DFG
    grants~KO~1053/7-2 and~VE~652/1-2, respectively.}}
\author{V.~Arvind\inst{1}, Johannes Köbler\inst{2}, Gaurav
  Rattan\inst{1}, Oleg Verbitsky\inst{2,}\inst{3}}

\institute{The Institute of Mathematical Sciences, Chennai 600 113,
  India \email{$\{$arvind,grattan$\}$@imsc.res.in} \and Institut f\"ur
  Informatik, Humboldt Universit\"at zu Berlin, Germany
  \email{$\{$koebler,verbitsk$\}$@informatik.hu-berlin.de} \and On
  leave from the Institute for Applied Problems of Mechanics and
  Mathematics, Lviv, Ukraine.}  \date{}
\begin{document} 

\maketitle

\begin{abstract}
  \ifcompact{}{\emph{Color refinement} is a classical technique used to show that
  two given graphs $G$ and $H$ are non-isomorphic; it is very
  efficient, although it does not succeed on all graphs. We call a
  graph $G$ \emph{amenable} to color refinement if the
  color-refinement procedure succeeds in distinguishing $G$ from any
  non-isomorphic graph $H$.} \ifcolref{Babai, Erd\H{o}s, and Selkow (1982) have
  shown that random graphs are amenable with high probability.

}{%

  Tinhofer (1991) explored a linear programming approach to Graph
  Isomorphism and defined the notion of \emph{compact graphs}: A graph
  is \emph{compact} if its fractional automorphisms polytope is
  integral. Tinhofer noted that isomorphism testing for compact graphs
  can be done quite efficiently by linear programming.  However, the
  problem of characterizing and recognizing compact graphs in
  polynomial time remains an open question. \ifcompact{In this paper we
    make new progress in our understanding of compact graphs.}{} Our
  results are summarized below:
\begin{itemize}[leftmargin=0mm,labelsep=1mm,itemindent=4mm,parsep=.5mm]}%
\ifcompact{}{\item We determine the exact range of applicability of color
  refinement by showing that amenable graphs are recognizable in time
  $O((n+m)\log n)$, where $n$ and $m$ denote the number of vertices
  and the number of edges in the input graph.} \ifcolref{}{%
\item We show that all \ifcompact{graphs $G$ which are distinguishable
    from any non-isomorphic graph by the classical color-refinement
    procedure}{amenable graphs} are compact. Thus, the applicability
  range for Tinhofer's linear programming approach to isomorphism
  testing is at least as large as for the combinatorial approach based
  on color refinement.
\item Exploring the relationship between color refinement and
  compactness further, we study related combinatorial and algebraic
  graph properties introduced by Tinhofer and Godsil. We show that the
  corresponding classes of graphs form a hierarchy and we prove that
  recognizing each of these graph classes is \p-hard. In particular,
  this gives a first complexity lower bound for recognizing compact graphs.
\end{itemize}
}
\end{abstract}

\ifcolref{\clearpage\setcounter{page}{1}}{}

\section{Introduction}\label{sec:intro}

\ifcompact{}{%
The well-known \emph{color refinement} (also known as \emph{naive
  vertex classification}) procedure for Graph Isomorphism works as
follows: it begins with a uniform coloring of the vertices of two
graphs $G$ and $H$ and refines the vertex coloring step by step. In a
refinement step, if two vertices have identical colors but differently
colored neighborhoods (with the multiplicities of colors counted),
then these vertices get new different colors. The procedure terminates
when no further refinement of the vertex color classes is
possible. Upon termination, if the multisets of vertex colors in $G$
and $H$ are different, we can correctly conclude that they are not
isomorphic. However, color refinement sometimes fails to distinguish
non-isomorphic graphs. The simplest example is given by any two
non-isomorphic regular graphs of the same degree with the same number
of vertices.  Nevertheless, color refinement turns out to be a useful
tool not only in isomorphism testing but also in a number of other
areas; see \cite{GroheKMS14,KerstingMGG14,ShervashidzeSLMB11} and
references there.

For which pairs of graphs $G$ and $H$ does the color refinement
procedure succeed in solving Graph Isomorphism? Mainly this question
has motivated the study of color refinement from different
perspectives.

Immerman and Lander \cite{ImmermanL90}, in their highly influential
paper, established a close connection between color refinement and
2-variable first-order logic with counting quantifiers. They show that
color refinement distinguishes $G$ and $H$ if and only if
these graphs are distinguishable by a sentence in this logic.

A well-known approach to tackling intractable optimization problems is
to consider an appropriate linear programming relaxation. A similar
approach to isomorphism testing, based on the notion of a fractional
isomorphism (introduced by Tinhofer \cite{Tinhofer86} using the term
doubly stochastic isomorphism), turns out to be equivalent to color
refinement. Building on Tinhofer's work \cite{Tinhofer86}, it is shown
by Ramana, Scheinerman and Ullman \cite{RamanaSU94} (see also
Godsil~\cite{Godsil97}) that two graphs are indistinguishable by color
refinement if and only if they are fractionally isomorphic.

We say that color refinement \emph{applies} to a graph $G$ if it
succeeds in distinguishing $G$ from any non-isomorphic $H$. A graph to
which color refinement applies is called \emph{amenable}.  There are
interesting classes of amenable graphs:

\begin{enumerate}
\item An obvious class of graphs to which color refinement is
  applicable is the class of \emph{unigraphs}. Unigraphs are graphs
  that are determined up to isomorphism by their degree sequences;
  see, e.g., \cite{BorriCP11,Tyshkevich00}.
\item Trees are amenable (Edmonds \cite{BusackerS65,Valiente02}).
\item It is easy to see that all graphs for which the color refinement
  procedure terminates with all singleton color classes (i.e.\ the
  color classes form the discrete partition) are amenable. Babai,
  Erd{\"o}s, and Selkow \cite{BabaiES80} have shown that a random
  graph $G_{n,\scriptscriptstyle 1/2}$ has this property with high
  probability. Moreover, the discrete partition of
  $G_{n,\scriptscriptstyle 1/2}$ is reached within at most two
  refinement steps. This implies that graph isomorphism is solvable
  very efficiently in the average case (see~\cite{BK79}).
\end{enumerate}}

\ifcompact{%
Consider the following natural linear programming
formulation \cite{Tinhofer86,RamanaSU94} of Graph Isomorphism. Let $G$
and $H$ be two graphs on $n$ vertices with adjacency matrices $A$ and
$B$, respectively. Then $G$ and $H$ are isomorphic if and only if
there is an $n\times n$ permutation matrix $X$ such that
  $AX=XB$.
A linear programming relaxation of this system of equations is to
allow $X$ to be a doubly stochastic matrix. 
If such an $X$ exists, it is called a \emph{fractional isomorphism}
from $G$ to $H$, and these graphs are said to be \emph{fractionally isomorphic}.

It turns out remarkably, as shown by \cite{RamanaSU94}, that two
graphs $G$ and $H$ are fractionally isomorphic if and only if they are
indistinguishable by the classical color-refinement procedure
(we outline this approach to isomorphism testing in Section
\ref{sec:compact}).

For which class of graphs is testing fractional isomorphism equivalent
to isomorphism testing? We call a graph $G$ \emph{amenable} if for any
other graph $H$ it holds that $G$ and $H$ are fractionally isomorphic
exactly when $G$ and $H$ are isomorphic.

The characterization of pairs of fractionally isomorphic graphs in terms of
color refinement implies that amenable graphs include
some well-known graph classes with easy isomorphism testing
algorithms, e.g.\ unigraphs, 
trees, and graphs for which color
refinement terminates with singleton color classes. 
We call graphs with the last property \emph{discrete}.
Babai, Erd{\"o}s, and Selkow \cite{BabaiES80}
have shown that a random graph $G_{n,\scriptscriptstyle 1/2}$ 
is discrete with high probability.
Thus, almost all graphs are amenable, which makes graph
isomorphism efficiently solvable in the average case.
Recently, we showed that the class of \emph{all} amenable graphs
can be recognized in nearly linear time \cite{arxiv};
a similar result was obtained independently in \cite{KSS15}.
This reduces testing of isomorphism for $G$ and $H$ to
computing a fractional isomorphism between these graphs,
once $G$ passes the amenability tests of~\cite{arxiv,KSS15}.
}{}

\ifcolref{}{%
The concept of a fractional isomorphism was used by Tinhofer
in \cite{Tinhofer91} as a basis for yet another 
linear-programming approach to isomorphism testing.
Tinhofer calls a graph $G$ \emph{compact} if the polytope of all its
fractional automorphisms is integral; more precisely, if $A$ is the
adjacency matrix of $G$, then the polytope in $\reals^{n^2}$
consisting of the doubly stochastic matrices $X$ such that
$AX=XA$ has only integral extreme points (i.e.\ all
coordinates of these points are integer).

If a compact graph $G$ is isomorphic to another graph $H$, then the
polytope of fractional isomorphisms from
$G$ to $H$ is also integral. If $G$ is not isomorphic to $H$, then
this polytope has \emph{no} integral extreme point. Thus, isomorphism
testing for compact graphs can be done in polynomial time by using
linear programming to compute an extreme point of the polytope and
testing if it is integral.

Before testing isomorphism of given $G$ and $H$ in this way,
we need to know that $G$ is compact. Unfortunately,
no efficient characterization of these graphs is currently known.
}

\ifcompact{}{%
\subsection*{Our results}

What is the class of graphs to which color refinement applies?  The
logical and linear programming based characterizations of color
refinement do not provide any efficient criterion answering this question. 

We aim at determining the exact range of applicability of color
refinement. We find an efficient characterization of the entire class
of amenable graphs, which allows for a quasilinear-time test whether
or not color refinement applies to a given graph. This result is shown
in Section~\ref{ss:exampl}, after we unravel the structure of amenable
graphs in Sections~\ref{ss:local} and~\ref{ss:global}.  We note that a
weak \textit{a priori} upper bound for the complexity of recognizing
amenable graphs is $\conp^{\gi[1]}$, where the superscript means the
one-query access to an oracle solving the graph isomorphism
problem. To the best of our knowledge, no better upper bound was known
before.

Combined with the Immerman-Lander result \cite{ImmermanL90} mentioned
above, it follows that the class of graphs definable by first-order
sentences with 2 variables and counting quantifiers is recognizable in
polynomial time.
}

\ifcolref{}{%
As our \ifcompact{}{second }main result, in Sections \ref{sec:compact} and 
\ref{sec:proofcompact} we show that 
 all amenable graphs are compact.
 Thus, Tinhofer's approach to Graph Isomorphism \cite{Tinhofer91} has
 at least as large an applicability range as color refinement.  In
 fact, the former approach is even more powerful because it is known
 that the class of compact graphs contains many regular graphs (for
 example, all cycles \cite{Tinhofer86}), for which no nontrivial color
 refinement is possible.

 In Section \ref{s:hierarchy}, we look at the relationship between the
 concepts of compactness and color refinement also from the other
 side.  Let us call a graph $G$ \emph{refinable} if the color
 partition produced by color refinement coincides with the orbit
 partition of the automorphism group of $G$.  It is interesting to
 note that the color-refinement procedure gives an efficient algorithm
 to check if a given refinable graph has a nontrivial automorphism.
 It follows from the results in \cite{Tinhofer91} that all compact
 graphs are refinable. The inclusion $\amenable\subset\compact$,
 therefore, implies that all amenable graphs are refinable as well.
 The last result is independently obtained in \cite{KSS15} by a
 different argument.  In the particular case of trees, this fact was
 observed long ago by several authors; see a survey
 in~\cite{TinhoferK99}.

Taking a finer look at the inclusion $\compact\subset\refinable$, 
in Section \ref{s:hierarchy} we discuss
algorithmic and algebraic graph properties that were introduced by
Tinhofer \cite{Tinhofer91} and Godsil \cite{Godsil97}. 
We note that, along with the other graph classes under consideration,
the corresponding classes \tinhofer and \godsil
form a hierarchy under inclusion:
\begin{eqnarray}\label{eq:hier-classes}
\discrete\subset \amenable\subset\compact\subset\godsil\subset \tinhofer\subset\refinable.
\end{eqnarray}
We show the following results on these graph classes:
\begin{itemize}
\item[$\bullet$] The hierarchy \refeq{hier-classes} is strict.
\item[$\bullet$] Testing membership in any of these
graph classes is \p-hard.
\end{itemize}

We prove the last fact by giving a suitable uniform AC$^0$ many-one
reduction from the \p-complete monotone boolean circuit-value problem
MCVP. More precisely, for a given MCVP instance $(C,x)$ our reduction
outputs a vertex-colored graph $G_{C,x}$ such that if $C(x)=1$ then
$G_{C,x}$ is discrete and if $C(x)=0$ then $G_{C,x}$ is not refinable.
In particular, the graph classes $\discrete$ and $\amenable$ are
P-complete. We note that Grohe \cite{Grohe99} established, for each
$k\ge2$, the \p-completeness of the equivalence problem for
first-order $k$-variable logic with counting quantifiers;
according to \cite{ImmermanL90}, this implies the \p-completeness of 
indistinguishability of two input graphs by color refinement.
We adapt the gadget constructions in \cite{Grohe99} to show our \p-hardness results.

}

\ifcompact{%
\paragraph{Related work.}
Particular families of compact graphs were identified in
\cite{Brualdi88,Godsil97,SchreckT88,WangLi05}; see also Chapter 9.10
in the monograph \cite{Brualdi06}. The concept of compactness is
generalized to \emph{weak compactness} in
\cite{EvdokimovKP99,EvdokimovPT00}. 

The linear programming approach of \cite{Tinhofer86,RamanaSU94} to
isomorphism testing is extended in
\cite{AtseriasM13,GroheO12,Malkin14}, where it is shown that this extension
corresponds to the $k$-dimensional Weisfeiler-Leman algorithm (which
is just color refinement if $k=1$).
}{%
\paragraph{Note.}  At the same time as our result appeared in an e-print
\cite{arxiv}, Sandra Kiefer, Pascal Schweitzer, and Erkal Selman
announced independently a similar characterization of amenable graphs
that subsequently appeared as an e-print \cite{KSS15}.
%

}

\paragraph{Notation.}
The vertex set of a graph $G$ is denoted by $V(G)$.  The vertices
adjacent to a vertex $u\in V(G)$ form its neighborhood $N(u)$.  A set
of vertices $X\subseteq V(G)$ induces a subgraph of $G$, that is
denoted by $G[X]$. For two disjoint sets $X$ and $Y$, $G[X,Y]$ is the
bipartite graph with vertex classes $X$ and $Y$ formed by all edges of
$G$ connecting a vertex in $X$ with a vertex in $Y$.  The
vertex-disjoint union of graphs $G$ and $H$ will be denoted by
$G+H$. Furthermore, we write $mG$ for the disjoint union of $m$ copies
of $G$.  The \emph{bipartite complement} of a bipartite graph $G$ with
vertex classes $X$ and $Y$ is the bipartite graph $G'$ with the same
vertex classes such that $\{x,y\}$ with $x\in X$ and $y\in Y$ is an
edge in $G'$ if and only if it is not an edge in $G$.  We use the
standard notation $K_n$ for the complete graph on $n$ vertices,
$K_{s,t}$ for the complete bipartite graph whose vertex classes have
$s$ and $t$ vertices, and $C_n$ for the cycle on $n$ vertices.

\ifcompact{\ifproceedings{%
In this extended abstract all missing proofs are given in the
appendix.
}{}}{}%

\ifcompact{%
\section{Amenable graphs}
}{%
\section{Basic definitions and facts}
\label{ss:amenable-basics}
}

For convenience, we will consider graphs to be vertex-colored in the
paper. A \emph{vertex-colored graph} is an undirected
simple graph $G$ endowed with a vertex coloring $c\function{V(G)}{\{1,\ldots,k\}}$. Automorphisms of
a vertex-colored graph and isomorphisms between vertex-colored graphs
are required to preserve vertex colors. We get usual graphs when $c$
is constant.

Given a graph $G$, the \emph{color-refinement} algorithm (to be
abbreviated as \emph{CR}) iteratively computes a sequence of colorings
$C^i$ of $V(G)$.  The initial coloring $C^0$ is the vertex coloring of
$G$, i.e., $C^0(u)=c(u)$. Then,
\begin{equation}
  \label{eq:Ci}
 C^{i+1}(u)=(C^i(u),\msetdef{C^i(a)}{a\in N(u)}),
\end{equation}
where $\left\{\!\!\left\{ \ldots \right\}\!\!\right\}$ denotes a
multiset. 



The partition $\Pa^{i+1}$ of $V(G)$ into the color classes of
$C^{i+1}$ is a refinement of the partition $\Pa^i$ corresponding to
$C^i$.  It follows that, eventually, $\Pa^{s+1}=\Pa^s$ for some $s$;
hence, $\Pa^{i}=\Pa^s$ for all $i\ge s$.  The partition $\Pa^s$ is
called the \emph{stable partition} of $G$ and denoted by~$\Pa_G$.

\ifcompact{%
Since the colorings $C^i$ are preserved under isomorphisms, for
isomorphic $G$ and $H$ we always have the equality
}{%
Given a partition $\Pa$ of the vertex set of a graph $G$, we call its
elements \emph{cells}.  We call $\Pa$ \emph{equitable} if:
\begin{enumerate}[leftmargin=8.5mm,label=(\roman*),topsep=1mm,itemsep=.5mm]
\item Each cell $X\in\Pa$ is monochromatic, i.e., all vertices $u,v\in
  X$ have the same color $c(u)=c(v)$.
\item For any cell $X\in\Pa$ the graph $G[X]$ induced by $X$ is
  \emph{regular}, that is, all vertices in $G[X]$ have equal degrees.
\item For any two cells $X,Y\in\Pa$ the bipartite graph $G[X,Y]$
  induced by $X$ and $Y$ is \emph{biregular}, that is, all vertices in
  $X$ have equally many neighbors in $Y$ and vice versa.
\end{enumerate}
%
It is easy to see that the stable partition of $G$ is equitable;
our analysis in the next section will make use of this fact.

A straightforward inductive argument shows that the colorings $C^i$
are preserved under isomorphisms.

\begin{lemma}\label{lem:CR-phi}
  If $\phi$ is an isomorphism from $G$ to $H$, then
  $C^i(u)=C^i(\phi(u))$ for any vertex $u$ of~$G$.
\end{lemma}

Lemma \ref{lem:CR-phi} readily implies that, if graphs $G$ and $H$ are
isomorphic, then
}
\begin{equation}
  \label{eq:CR-check}
 \msetdef{C^i(u)}{u\in V(G)} = \msetdef{C^i(v)}{v\in V(H)}
\end{equation}
for all $i\ge0$\ifcompact{ or, equivalently, for $i=2n$. The CR
  algorithm accepts two graphs $G$ and $H$ as isomorphic exactly under
  this condition. To avoid an exponential growth of the lengths of
  color names, CR renames the colors after each refinement
  step.

We call a graph $G$ \emph{amenable} if CR works correctly on the input
$G,H$ for every $H$, that is, Equality~\refeq{CR-check} is false for
$i=2n$ whenever $H\not\cong G$.

The elements of the stable partition $\Pa_G$ of a graph $G$
will be called \emph{cells}. 
We define
the auxiliary \emph{cell graph} $C(G)$ of $G$ to be the complete graph
on the vertex set $\Pa_G$.  A vertex $X$ of the cell graph is called
\emph{homogeneous} if the graph $G[X]$ is either {complete} or {empty}
and \emph{heterogeneous} otherwise.  An edge $\{X,Y\}$ of the cell
graph is called \emph{isotropic} if the bipartite graph $G[X,Y]$ is
either {complete} or {empty} and \emph{anisotropic} otherwise. By an
\emph{anisotropic component} of the cell graph $C(G)$ we mean a
maximal connected subgraph of $C(G)$ whose edges are all anisotropic.
Note that if a vertex of $C(G)$ has no incident anisotropic edges, it
forms a single-vertex anisotropic component.

\begin{theorem}[Arvind et al.~\cite{arxiv}]\label{thm:amenable}
  A graph $G$ is amenable if and only if the stable partition $\Pa_G$
  of~$G$ fulfils the following three properties:
\begin{enumerate}[leftmargin=9mm,label=\normalfont\bfseries(\Alph*)]
\item For any cell $X\in\Pa_G$, $G[X]$ is an empty graph, a complete
  graph, a matching graph $mK_2$, the complement of a matching graph,
  or the 5-cycle;
\item For any two cells $X,Y\in\Pa_G$, $G[X,Y]$ is an empty graph, a
  complete bipartite graph, a disjoint union of stars $sK_{1,t}$
where $X$ and $Y$ are the set of $s$
  central vertices and the set of $st$ leaves, or the bipartite
  complement of the last graph.
\item Every anisotropic component $A$ of the cell graph $C(G)$ of $G$ is a
  tree and contains at most one heterogeneous vertex.
If $A$ contains an heterogeneous vertex, it has
  minimum cardinality among the vertices of~$A$. Let $R$ be any vertex
  of $A$ having minimum cardinality and let $A_R$ be the rooted
  directed tree obtained from $A$ by rooting it at $R$.  Then
  $|X|\le|Y|$ for any directed edge $(X,Y)$ of~$A_R$.
  \end{enumerate}
\end{theorem}
}{. When used for isomorphism testing, the CR algorithm accepts
  two graphs $G$ and $H$ as isomorphic exactly when the above
  condition is met on input $G+H$. Note that this condition is
  actually finitary: If Equality~\refeq{CR-check} is false for some
  $i$, it must be false for some $i<2n$, where $n$ denotes the number
  of vertices in each of the graphs.  This follows from the
  observation that the partition $\Pa^{2n-1}$ induced by the coloring
  $C^{2n-1}$ must be the stable partition of the disjoint union of $G$
  and $H$. In fact, Equality~\refeq{CR-check} holds true for all $i$
  if it is true for $i=n$; see, e.g., \cite{KrebsV14}.  Thus, it is
  enough that CR verifies \refeq{CR-check} for $i=n$.

Note that computing the vertex colors literally according to
\refeq{Ci} would lead to an exponential growth of the lengths of color
names. This can be avoided by renaming the colors after each
refinement step. Then CR never needs more than $n$ color names
(appearance of more than $n$ colors is an indication that the graphs
are non-isomorphic).

\begin{definition}
  We call a graph $G$ \emph{amenable} if CR works correctly on the
  input $G,H$ for every $H$, that is, Equality~\refeq{CR-check} is
  false for $i=n$ whenever $H\not\cong G$.
\end{definition}

\section{Local structure of amenable graphs}\label{ss:local}

Consider the stable 
partition $\Pa_G$ of an amenable graph $G$.  The following lemma gives a list of
all possible regular and biregular graphs that can occur,
respectively, as $G[X]$ and $G[X,Y]$ for cells $X,Y$ of~$\Pa_G$.

\begin{restatable}{lemma}{necessaryAB}\label{lem:necessaryAB}
  The stable partition $\Pa_G$ of an amenable graph~$G$ fulfills the
  following properties:
\begin{enumerate}[leftmargin=9mm,label=\normalfont\bfseries(\Alph*),topsep=1mm,itemsep=.5mm]
\item For any cell $X\in\Pa_G$, $G[X]$ is an empty graph, a complete
  graph, a matching graph $mK_2$, the complement of a matching graph,
  or the 5-cycle;
\item For any two cells $X,Y\in\Pa_G$, $G[X,Y]$ is an empty graph, a
  complete bipartite graph, a disjoint union of stars $sK_{1,t}$
  where $X$ and $Y$ are the set of $s$ central vertices and the set of
  $st$ leaves, or the bipartite complement of the last graph.
\end{enumerate}
\end{restatable}

The proof of Lemma~\ref{lem:necessaryAB} is\ifproceedings{ given in the
  appendix. It is}{} based on the following facts.

\begin{lemma}[Johnson \cite{Johnson75}]\label{lem:johnson}
  A regular graph of degree $d$ with $n$ vertices is a unigraph if and
  only if $d\in\{0,1,n-2,n-1\}$ or $d=2$ and $n=5$.\footnote{%
    The last case, in which the graph is the 5-cycle, is missing from
    the statement of this result in \cite[Theorem
    2.12]{Johnson75}. The proof in \cite{Johnson75} tacitly considers
    only graphs with at least 6 vertices.}
\end{lemma}

\begin{lemma}[Koren \cite{Koren76}]\label{lem:koren}
  A bipartite graph is determined up to isomorphism by the conditions
  that every of the $m$ vertices in one part has degree $c$ and every
  of the $n$ vertices in the other part has degree $d$ if and only if
  $c\in\{0,1,n-1,n\}$ or $d\in\{0,1,m-1,m\}$.
\end{lemma}

If $G$ contains a subgraph $G[X]$ or $G[X,Y]$ that is induced by some
$X,Y\in\Pa_G$ but not listed in Lemma \ref{lem:necessaryAB}, then
Lemmas \ref{lem:johnson} and \ref{lem:koren} imply that this subgraph
can be replaced by a non-isomorphic regular or biregular graph with
the same parameters. Hence, in order to prove Lemma
\ref{lem:necessaryAB} it suffices to show that the resulting graph $H$
is indistinguishable from $G$ by color refinement.  The graphs $G$ and
$H$ in the following lemma have the same vertex set.  Given a vertex
$u$, we distinguish its neighborhoods $N_G(u)$ and $N_H(u)$ and its
colors $C^i_G(u)$ and $C^i_H(u)$ in the two graphs.

\begin{restatable}{lemma}{localsurgery}\label{lem:local-surgery}
  Let $X$ and $Y$ be cells of the stable partition of a graph~$G$.
\begin{enumerate}[leftmargin=6.5mm,label=(\roman*),topsep=1mm,itemsep=.5mm]
\item If $H$ is obtained from $G$ by replacing the edges of the
  subgraph $G[X]$ with the edges of an arbitrary regular graph $(X,E)$
  having the same degree, then $C^i_G(u)=C^i_H(u)$ for any $u\in V(G)$
  and any~$i$.
\item If $H$ is obtained from $G$ by replacing the edges of the
  subgraph $G[X,Y]$ with the edges of an arbitrary biregular graph
  with the same vertex partition such that the vertex degrees remain
  unchanged, then $C^i_G(u)=C^i_H(u)$ for any $u\in V(G)$ and any~$i$.
\end{enumerate}
\end{restatable}

\newcommand{\prooflocalsurgery}{%
\begin{proofof}{Lemma \ref{lem:local-surgery}}
  We proceed by induction on $i$. In the base case of $i=0$ the claim
  is trivially true. Assume that $C^i_G(a)=C^i_H(a)$ for all $a\in
  V(G)$.  We consider an arbitrary vertex $u$ and prove that
  \begin{equation}
    \label{eq:CiGH}
    C^{i+1}_G(u)=C^{i+1}_H(u). 
  \end{equation}
  {}From now on we treat Parts (i) and (ii) separately.

  \begin{enumerate}[leftmargin=0mm,labelsep=2mm,itemindent=6.5mm,label=(\roman*),topsep=1.5mm,itemsep=1.5mm]
  \setlength{\parindent}{5mm}
  \item Suppose first that $u\notin X$. Since the transformation of
    $G$ into $H$ does not affect the edges emanating from $u$, we have
    $N_G(u)=N_H(u)$.  Looking at the definition \refeq{Ci}, we
    immediately derive \refeq{CiGH} from the induction assumption.

    If $u\in X$, we only have the equality $N_G(u)\setminus
    X=N_H(u)\setminus X$, implying that
    \begin{equation}
      \label{eq:eq1}
      \msetdef{C^i_G(a)}{a\in N_G(u)\setminus X} = \msetdef{C^i_H(a)}{a\in N_H(u)\setminus X}.
    \end{equation}
    The equality $N_G(u)\cap X=N_H(u)\cap X$ is not necessarily true.
    However, $u$ has equally many neighbors from $X$ in $G$ and in
    $H$.  Furthermore, for any two vertices $a$ and $a'$ in $X$ we
    have $C^i_G(a)=C^i_G(a')$ because $X$ is a cell of $G$, and
    $C^i_H(a)=C^i_G(a)=C^i_G(a')=C^i_H(a')$ by the induction
    assumption.  That is, all vertices in $X$ have the same
    $C^i$-color both in $G$ and in $H$.  It follows that
    \begin{equation}
      \label{eq:eq2}
      \msetdef{C^i_G(a)}{a\in N_G(u)\cap X} = \msetdef{C^i_H(a)}{a\in N_H(u)\cap X}.
    \end{equation}
    Combining \refeq{eq1} and \refeq{eq2}, we conclude that
    \refeq{CiGH} holds in any case.
  \item If $u\notin X\cup Y$, we have $N_G(u)=N_H(u)$ and Equality
    \refeq{CiGH} readily follows from the induction assumption.

    Suppose that $u\in Y$. In this case we still have \refeq{eq1} and,
    exactly as in Part (i), we also derive \refeq{eq2}.  Equality
    \refeq{CiGH} follows.  The case of $u\in X$ is symmetric.\qed
  \end{enumerate}
\end{proofof}
}\ifproceedings{The proof of Lemma~\ref{lem:local-surgery} is given in
  the appendix.}{\prooflocalsurgery}%

\newcommand{\proofnecessaryAB}{%
\begin{proofof}{Lemma \ref{lem:necessaryAB}}
  \begin{enumerate}[leftmargin=0mm,labelsep=2mm,itemindent=8mm,label=\normalfont\bfseries(\Alph*),topsep=1.5mm,itemsep=1.5mm]
  \item If $G[X]$ is a graph not from the list, by Lemma
    \ref{lem:johnson}, it is not a unigraph. Hence, we can modify $G$
    locally on $X$ by replacing $G[X]$ with a non-isomorphic regular
    graph with the same parameters. Part (i) of Lemma
    \ref{lem:local-surgery} implies that the resulting graph $H$
    satisfies Equality \refeq{CR-check} for any $i$, implying that CR
    does not distinguish between $G$ and $H$. The graphs $G$ and $H$
    are non-isomorphic because, by Part (i) of Lemma
    \ref{lem:local-surgery} and by Lemma \ref{lem:CR-phi}, an
    isomorphism from $G$ to $H$ would induce an isomorphism from
    $G[X]$ to $H[X]$.  This shows that $G$ is not amenable.
  \item This condition follows, similarly to Condition~{\bfseries A},
    from Lemma \ref{lem:koren} and Part (ii) of
    Lemma~\ref{lem:local-surgery}.\qed
  \end{enumerate}
\end{proofof}
}\ifproceedings{\proofnecessaryAB}{\proofnecessaryAB}%

\section{Global structure of amenable graphs}\label{ss:global}

Recall that $\Pa_G$ is the stable partition of the vertex set of a
graph $G$, and that elements of $\Pa_G$ are called cells.  We define
the auxiliary \emph{cell graph} $C(G)$ of $G$ to be the complete graph
on the vertex set $\Pa_G$ with the following labeling of vertices and
edges. A vertex $X$ of $C(G)$ is called \emph{homogeneous} if the
graph $G[X]$ is either {complete} or {empty} and \emph{heterogeneous}
otherwise.  An edge $\{X,Y\}$ of $C(G)$ is called \emph{isotropic} if
the bipartite graph $G[X,Y]$ is either {complete} or {empty} and
\emph{anisotropic} otherwise. A path $X_1X_2\ldots X_l$ in $C(G)$
where every edge $\{X_i,X_{i+1}\}$ is anisotropic will be referred to
as an \emph{anisotropic path}.  If also $\{X_l,X_1\}$ is an
anisotropic edge, we speak of an \emph{anisotropic cycle}.  In the
case that $|X_1|=|X_2|=\ldots=|X_l|$, such a path (or cycle) is called
\emph{uniform}.

For graphs fulfilling Conditions~{\bfseries A} and~{\bfseries B} of
Lemma~\ref{lem:necessaryAB} we refine the labeling of the vertices and
edges of $C(G)$ as follows.  A heterogeneous cell $X\in\Pa_G$ is
called \emph{matching}, \emph{co-matching}, or \emph{pentagonal}
depending on the type of $G[X]$. Note that a matching or co-matching
cell $X$ always consists of at least 4 vertices. Further, an
anisotropic edge $\{X,Y\}$ is called \emph{constellation} if $G[X,Y]$
is a disjoint union of stars, and \emph{co-constellation} otherwise
(i.e., the bipartite complement of $G[X,Y]$ is a disjoint union of
stars). Likewise, homogeneous cells $X$ (and isotropic edges
$\{X,Y\}$) are called \emph{empty} if the graph $G[X]$
(resp. $G[X,Y]$) is empty, and \emph{complete} otherwise.


Note that if an edge $\{X,Y\}$ of a uniform path/cycle is
\emph{constellation} (resp.\ \emph{co-constellation}), then $G[X,Y]$
is a matching (resp.\ co-matching) graph.

\begin{restatable}{lemma}{necessaryCD}\label{lem:necessaryCD}
  The cell graph $C(G)$ of an amenable graph $G$ has the following
  properties:
\begin{enumerate}[leftmargin=8mm,label=\normalfont\bfseries(\Alph*),topsep=1mm,itemsep=.5mm]\setcounter{enumi}{2}
\item $C(G)$ contains no uniform anisotropic path connecting two
  heterogeneous cells;
\item $C(G)$ contains no uniform anisotropic cycle;
\item $C(G)$ contains neither an anisotropic path $XY_1\ldots Y_lZ$
  such that $|X|<|Y_1|=\ldots=|Y_l|>|Z|$ nor an anistropic cycle
  $XY_1\ldots Y_lX$ such that $|X|<|Y_1|=\ldots=|Y_l|$;
\item $C(G)$ contains no anisotropic path $XY_1\ldots Y_l$ such that
  $|X|<|Y_1|=\ldots=|Y_l|$ and the cell $Y_l$ is heterogeneous.
\end{enumerate}
\end{restatable}

\newcommand{\proofnecessaryCD}{%
\begin{proof}
  \begin{enumerate}[leftmargin=0mm,labelsep=2mm,itemindent=8mm,label=\normalfont\bfseries(\Alph*),topsep=1.5mm,itemsep=1.5mm]\setcounter{enumi}{2}
  \setlength{\parindent}{5mm}
  \item Suppose that $P$ is a uniform anisotropic path in $C(G)$
    connecting two heterogeneous cells $X$ and $Y$. Let
    $k=|X|=|Y|$. Complementing $G[A,B]$ for each co-constellation edge
    $\{A,B\}$ of $P$, in $G$ we obtain $k$ vertex-disjoint paths
    connecting $X$ and $Y$.  These paths determine a one-to-one
    correspondence between $X$ and $Y$.  Given $v\in X$, denote its
    mate in $Y$ by $v^*$.  Call $P$ \emph{conducting} if this
    correspondence is an isomorphism between $G[X]$ and $G[Y]$, that
    is, two vertices $u$ and $v$ in $X$ are adjacent exactly when
    their mates $u^*$ and $v^*$ are adjacent. In the case that one of
    $X$ and $Y$ is matching and the other is co-matching, we call $P$
    \emph{conducting} also if the correspondence is an isomorphism
    between $G[X]$ and the complement of~$G[Y]$.

    We construct a non-isomorphic graph $H$ such that CR does not
    distinguish between $G$ and $H$. Since $Y$ is heterogeneous, we
    can replace the edges of the subgraph $G[Y]$ with the edges of an
    isomorphic but different subgraph $(Y,E)$. Since also $X$ is
    heterogeneous it follows that $P$ is a conducting path in the
    resulting graph $H$ if and only if $P$ is a non-conducting path in
    $G$. Now, Part~(i) of Lemma \ref{lem:local-surgery} implies that
    CR computes the same stable partition for $G$ and $H$ and does not
    distinguish between them. On the other hand,
    Lemma~\ref{lem:CR-phi} implies that any isomorphism $\phi$ between
    $G$ and $H$ must map each cell to itself.  As $\phi$ must also
    preserve the conducting property along the path $P$, it follows
    that $G$ and $H$ are not isomorphic. Hence, $G$ is not amenable.

  \item Suppose that $C(G)$ contains a uniform anisotropic cycle $Q$
    of length $m$.  All cells in $Q$ have the same cardinality; denote
    it by $k$.  Complementing $G[A,B]$ for each co-constellation edge
    $\{A,B\}$ of $Q$, in $G$ we obtain the vertex-disjoint union of
    cycles whose lengths are multiples of $m$. As two extreme cases,
    we can have $k$ cycles of length $m$ each or we can have a single
    cycle of length $km$.  Denote the isomorphism type of this union
    of cycles by $\tau(Q)$.  Note that this type is isomorphism
    invariant: For an isomorphism $\phi$ from $G$ to another graph
    $H$, $\tau(\phi'(Q))=\tau(Q)$ for the induced isomorphism $\phi'$
    from $C(G)$ to $C(H)$.

    Let $X$ and $Y$ be two consecutive cells in $Q$. We can replace
    the subgraph $G[X,Y]$ with an isomorphic but different bipartite
    graph so that in the resulting graph $H$, $\tau(Q)$ becomes either
    $kC_m$ or $C_{km}$, whatever we wish. In particular, we can
    replace the subgraph $G[X,Y]$ in such a way that $\tau(Q)$ is
    changed.

    Similarly as for Condition~{\bfseries C}, we use Part (ii) of
    Lemma \ref{lem:local-surgery} to argue that CR does not
    distinguish between $G$ and $H$. Furthermore, $G\not\cong H$
    because the types $\tau(Q)$ in $G$ and $H$ are
    different. Therefore, $G$ is not amenable.
  \item Suppose that $C(G)$ contains an anisotropic path $P=XY_1\ldots
    Y_lZ$ such that $|X|<|Y_1|=\ldots=|Y_l|>|Z|$ (for the case of a
    cycle, where $Z=X$, the argument is virtually the same).  Let
    $G[X,Y_1]=sK_{1,t}$ and $G[Z,Y_l]=aK_{1,b}$, where $s,a,t,b\ge2$
    (if any of these subgraphs is a co-constellation, we consider its
    complement). Thus, $|X|=s$, $|Z|=a$, and $|Y_1|=|Y_l|=st=ab$.

    Like in the proof of Condition~{\bfseries C}, the uniform
    anisotropic path $Y_1\ldots Y_l$ determines a one-to-one
    correspondence between the cells $Y_1$ and $Y_l$ that can be used
    to make the identification $Y_1=Y_l=\{1,2,\ldots,st\}=Y$. For each
    $x\in X$, let $Y_x$ denote the set of vertices in $Y$ adjacent to
    $x$. The set $Y_z$ is defined similarly for each $z\in Z$. Note
    that for any $x\ne x'$ in $X$ and $z\ne z'$ in $Z$,
    \begin{equation*}
      |Y_x|=t,\ \ |Y_z|=b,\ \ Y_x\cap Y_{x'}=\emptyset, 
      \text{ \ and \ }Y_z\cap Y_{z'}=\emptyset.
    \end{equation*}
    We regard $\calY_G=\Set{Y_x}_{x\in X}\cup\Set{Y_z}_{z\in Z}$ as a
    hypergraph on the vertex set $Y$. Note that $\calY_G$ might be a
    multi-hypergraph as the two hyperedges $Y_x$ and $Y_z$ might
    coincide for some pairs $(x,z)\in X\times Z$. Without loss of
    generality, we can assume that the hyperedges $Y_x$, $x\in X$,
    form consecutive intervals in $Y$. We call the anisotropic path
    $P$ \emph{flat}, if there exists no pair $(x,z)\in X\times Z$ such
    that one of the two hyperedge $Y_x$ and $Y_z$ is contained in the
    other.

    We construct a non-isomorphic graph $H$ such that CR does not
    distinguish between $G$ and $H$. If $P$ is flat in $G$, we replace
    the edges of the subgraph $G[Z,Y_l]$ by the edges of an isomorphic
    but different biregular graph such that $P$ becomes non-flat in
    the resulting graph $H$. More precisely, we replace the edges in
    such a way that all hyperedges of $\calY_H$ form consecutive
    intervals in $Y$ by letting $\calY_H=\Set{Y_x}_{x\in
      X}\cup\Set{Y_{i}}_{i\in[a]}$, where
    $Y_i=\{(i-1)b+1,\ldots,ib\}$. Likewise, if $P$ is non-flat in $G$,
    we replace the edges of $G[Z,Y_l]$ such that $P$ becomes flat in
    $H$ by letting $\calY_H=\Set{Y_x}_{x\in
      X}\cup\Set{Y_i}_{i\in[a]}$, where $Y_i=\{i,i+a,i+(b-1)a\}$.

    Now, Part~(i) of Lemma \ref{lem:local-surgery} implies that CR
    computes the same stable partition for $G$ and $H$ and does not
    distinguish between them. On the other hand,
    Lemma~\ref{lem:CR-phi} implies that any isomorphism $\phi$ between
    $G$ and $H$ must map each cell to itself.  As $\phi$ must also
    preserve the flatness property of the path $P$, it follows
    that $G$ and $H$ are not isomorphic. Hence, $G$ is not amenable.
  \item Suppose that $C(G)$ contains an anisotropic path $XY_1\ldots
    Y_l$ where $|X|<|Y_1|=\ldots=|Y_l|$ and $Y_l$ is heterogeneous.
    Let $G[X,Y_1]=sK_{1,t}$ (in the case of a co-constellation, we
    consider the complement).  Since $s,t\ge2$ and $|Y_1|=st$, the
    cell $Y_l$ cannot be pentagonal. Considering the complement if
    needed, we can assume without loss of generality that $Y_l$ is
    matching.  Like in the proof of Condition~{\bfseries E}, the
    uniform anisotropic path $Y_1\ldots Y_l$ determines a one-to-one
    correspondence between the cells $Y_1$ and $Y_l$ that can be used
    to make the identification $Y_1=Y_l=\{1,2,\ldots,st\}=Y$.
    Consider the hypergraph $\calY_G=\Set{Y_x}_{x\in X}\cup
    E(G[Y_l])$, where $Y_x=N_G(x)\cap Y_1$ and $E(G[Y_l])$ denotes the
    edge set of $G[Y_l]$.  Now, exactly as in the proof of
    Condition~{\bfseries E}, we can change the isomorphism type of
    $\calY_G$ by replacing the edges of the subgraph $G[X,Y_1]$ by the
    edges of an isomorphic biregular graph. This yields a
    non-isomorphic graph $H$ that is indistinguishable from $G$ by
    CR.\qed
  \end{enumerate}
\end{proof}
}{\proofnecessaryCD}

It turns out that Conditions~{\bfseries A--F} are not only necessary
for amenability (as shown in Lemmas \ref{lem:necessaryAB} and
\ref{lem:necessaryCD}) but also sufficient. As a preparation we first
prove the following Lemma \ref{lem:aniso-comp} that reveals a
tree-like structure of amenable graphs.  By an \emph{anisotropic
  component} of the cell graph $C(G)$ we mean a maximal connected
subgraph of $C(G)$ whose edges are all anisotropic.  Note that if a
vertex of $C(G)$ has no incident anisotropic edges, it forms a
single-vertex anisotropic component.

\begin{restatable}{lemma}{anisocomp}\label{lem:aniso-comp}
  Suppose that a graph $G$ satisfies Conditions~{\normalfont\bfseries
    A--F}.  Then for any anisotropic component $A$ of $C(G)$, the
  following is true.
  \begin{enumerate}[leftmargin=8mm,label=\normalfont\bfseries(\Alph*),topsep=1mm,itemsep=.5mm]\setcounter{enumi}{6}
  \item $A$ is a tree with the following monotonicity property.  Let
    $R$ be a cell in $A$ of minimum cardinality and let $A_R$ be the
    rooted directed tree obtained from $A$ by rooting $A$ at $R$.
    Then $|X|\le|Y|$ for any directed edge $(X,Y)$ of~$A_R$.
  \item $A$ contains at most one heterogeneous vertex.  If $R$ is such
    a vertex, it has minimum cardinality among the cells of~$A$.
  \end{enumerate}
\end{restatable}

\newcommand{\proofanisocomp}{%
\begin{proof}
  \begin{enumerate}[leftmargin=0mm,labelsep=2mm,itemindent=8mm,label=\normalfont\bfseries(\Alph*),topsep=1.5mm,itemsep=1.5mm]\setcounter{enumi}{6}
  \setlength{\parindent}{5mm}
  \item $A$ cannot contain any uniform cycle by Condition~{\bfseries
      D} and any other cycle by Condition~{\bfseries E}. The
    monotonicity property follows from Condition~{\bfseries E}.
  \item Assume that $A$ contains more than one heterogeneous cell.
    Consider two such cells $S$ and $T$.  Let
    $S=Z_1,Z_2,\ldots,Z_l=T$ be the path from $S$ to $T$ in $A$.  The
    monotonicity property stated in Condition~{\bfseries G} implies that there is $j$
    (possibly $j=1,l$) such that
    $|Z_1|\ge\ldots\ge|Z_j|\le\ldots\le|Z_l|$.  Since the path cannot
    be uniform by Condition~{\bfseries C}, at least one of the
    inequalities is strict.  However, this contradicts
    Condition~{\bfseries F}.

    Suppose that $R$ is a heterogeneous cell in $A$.  Consider now a
    path $R=Z_1,Z_2,\ldots,Z_l=S$ in $A$ where $S$ is a cell with
    the smallest cardinality.  By the monotonicity property and
    Condition~{\bfseries F}, this path must be uniform, proving that
    $|R|=|S|$.\qed
  \end{enumerate}
\end{proof}
}{\proofanisocomp}

In combination with Conditions~{\bfseries A} and {\bfseries B},
Conditions~{\bfseries G} and {\bfseries H} on anisotropic components give a very stringent
characterization of amenability.

\begin{restatable}{theorem}{anisocomptwo}\label{thm:amenable}
  For a graph $G$ the following conditions are equivalent:
  \begin{enumerate}[leftmargin=8.5mm,label=(\roman*),topsep=1mm,itemsep=.5mm]
  \item $G$ is amenable.
  \item $G$ satisfies Conditions~{\normalfont\bfseries A--F}.
  \item $G$ satisfies Conditions~{\normalfont\bfseries A, B, G}
    and~{\bfseries H}.
  \end{enumerate}
\end{restatable}

\newcommand{\proofanisocomptwo}{%
  \begin{proof} It only remains to show that any graph $G$ fulfilling
    the Conditions~{\normalfont\bfseries A, B, G} and~{\bfseries H} is
    amenable. Let $H$ be a graph indistinguishable from $G$ by
    CR. Then we have to show that $G$ and $H$ are isomorphic.

    Consider the coloring $C^s$ corresponding to the stable partition
    $\Pa^s$ of the disjoint union $G+H$.  Since $G$ and $H$ satisfy
    Equality~\refeq{CR-check} for $i=s$, there is a bijection
    $f:\Pa_G\to\Pa_H$ matching each cell $X$ of the stable partition
    of $G$ to the cell $f(X)\in\Pa_H$ such that the vertices in $X$
    and $f(X)$ have the same $C^s$-color.  Moreover,
    Equality~\refeq{CR-check} implies that $|X|=|f(X)|$.  We claim
    that for any cells $X$ and $Y$ of~$G$,
  \begin{enumerate}[leftmargin=6.5mm,label=(\alph*),topsep=1mm,itemsep=.5mm]
  \item $G[X]\cong H[f(X)]$ and
  \item $G[X,Y]\cong H[f(X),f(Y)]$,
  \end{enumerate}
  implying that $f$ is an isomorphism from $C(G)$ to $C(H)$. 

Indeed, since $X$ and $f(X)$ are cells of the stable partitions
$\Pa_G$ and $\Pa_H$, both $G[X]$ and $H[f(X)]$ are regular.
Since $X\cup f(X)$ is a cell of the stable
  partition $\Pa^s$ of $G+H$, the graphs $G[X]$ and $H[f(X)]$ have the same degree.  By
  Condition~{\bfseries A}, $G[X]$ is a unigraph, implying 
  Property (a). Property (b) follows from Condition~{\bfseries
    B} by a similar argument.

  We now construct an isomorphism $\phi$ from $G$ to $H$.  By Lemma
  \ref{lem:CR-phi}, we should have $\phi(X)=f(X)$ for each cell $X$.
  Therefore, we have to define the map $\phi\function X{f(X)}$ on
  each~$X$.

  By Condition~{\bfseries H}, an anisotropic component $A$ of the cell graph $C(G)$
  contains at most one heterogeneous cell.  Denote it by $R_A$ if it
  exists. Otherwise fix $R_A$ to be an arbitrary cell of the minimum cardinality in~$A$.

  For each $A$, define $\phi$ on $R=R_A$ to be an arbitrary
  isomorphism from $G[R]$ to $H[f(R)]$, which exists according to
  (a). After this, propagate $\phi$ to any other cell in $A$ as
  follows.  By Condition~{\bfseries G}, $A$ is a tree. Let $A_R$ be the directed
  rooted tree obtained from $A$ by rooting it at $R$. Suppose that
  $\phi$ is already defined on $X$ and $(X,Y)$ is an edge in $A$.
By the monotonicity property in Condition~{\bfseries G} and our choice of $R$,
we can assume that $|X|\le|Y|$.
  Then $\phi$ can be extended to $Y$ so that this is an isomorphism from
  $G[X,Y]$ to $H[f(X),f(Y)]$. This is possible by (b) due to the fact
that all vertices in $Y$ have degree 1 in $G[X,Y]$ or its bipartite complement 
(and the same holds for all vertices in $f(Y)$ in the graph $H[f(X),f(Y)]$).

  It remains to argue that the map $\phi$ obtained in this way is
  indeed an isomorphism from $G$ to $H$.  It suffices to show that
  $\phi$ is an isomorphism between $G[X]$ and $H[f(X)]$ for each cell
  $X$ of $G$ and between $G[X,Y]$ and $H[f(X),f(Y)]$ for each pair of
  cells $X$ and~$Y$.

  If $X$ is homogeneous, $f(X)$ is homogeneous of the same type,
  complete or empty, according to (a).  In this case, any $\phi$ is an
  isomorphism from $G[X]$ to $H[f(X)]$.  If $X$ is heterogeneous, the
  assumption of the lemma says that it belongs to a unique anisotropic
  component $A$ (and $X=R_A$).  Then $\phi$ is an isomorphism from
  $G[X]$ to $H[f(X)]$ by construction.

  If $\{X,Y\}$ is an isotropic edge of $C(G)$, then (b) implies that
  $\{f(X),f(Y)\}$ is an isotropic edge of $C(H)$ of the same type,
  complete or empty.  In this case, $\phi$ is an isomorphism from
  $G[X,Y]$ to $H[f(X),f(Y)]$, no matter how it is defined. If
  $\{X,Y\}$ is anisotropic, it belongs to some anisotropic component
  $A$, and $\phi$ is an isomorphism from $G[X,Y]$ to $H[f(X),f(Y)]$ by
  construction.\qed
\end{proof}
}{\proofanisocomptwo}


\vspace{-2.5mm}

\section{Examples and applications}\label{ss:exampl}

\vspace{-1.5mm}

Theorem \ref{thm:amenable} is a convenient tool for verifying
amenability.  For example, amenability of discrete graphs is a
well-known fact. Recall that those are graphs whose stable partitions
consist of singletons. As each cell is a singleton, any anisotropic
component of a discrete graph consists of a single cell. Hence,
Conditions~{\bfseries A} and {\bfseries B} as well as
Conditions~{\bfseries G} and {\bfseries H} on anisotropic components
are fulfilled by trivial reasons.

Checking these four conditions, we can also reprove the amenability of
trees.  Moreover, we can extend this result to the class of
forests. This extension does not seem to be straightforward because
the class of amenable graphs is not closed under disjoint unions.  For
example, $C_3+C_4$ is indistinguishable by CR from $C_7$ and, hence,
is not amenable.

\begin{restatable}{corollary}{forests}\label{cor:forests}
  All forests are amenable.
\end{restatable}

\begin{proof}
  A regular acyclic graph is either an empty or a matching graph. This
  implies Condition~{\bfseries A}.  Condition~{\bfseries B} follows
  from the observation that biregular acyclic graphs are either empty
  or disjoint unions of stars.

  Let $C^*(G)$ be the version of the cell graph $C(G)$ where all empty
  edges are removed. 
  If $C^*(G)$ contains a cycle, $G$ must contain a cycle as well.
  Therefore, if $G$ is acyclic, then $C^*(G)$ is acyclic too, and any
  anisotropic component of $C(G)$ must be a tree. To prove the
  monotonicity property in Condition~{\bfseries G}, it suffices to
  show that $C(G)$ cannot contain an anisotropic path $XY_1\ldots
  Y_lZ$ with $|X|<|Y_1|=\dots=|Y_l|>|Z|$. But this easily follows
  since in this case each vertex of the induced subgraph $G[X\cup
  Y_1\cup\ldots\cup Y_l\cup Z]$ has degree at least 2 in $G$,
  contradicting the acyclicity of $G$.

  To prove Condition~{\bfseries H}, suppose that $C(G)$ contains an
  anisotropic path $X_0,X_1,\ldots,X_l$ connecting two heterogeneous
  cells $X_0$ and $X_l$. Then each vertex of the induced subgraph
  $G[X_0\cup X_1\cup\ldots\cup X_{l-1}\cup X_l]$ has degree at least 2
  in $G$, a contradiction.  The same contradiction arises if such a
  path connects a heterogeneous cell $X_0$ with an arbitrary cell
  $X_l$, where $|X_l|<|X_{l-1}|$. Hence, $X_0$ must have minimum
  cardinality among all cells belonging to the same anisotropic
  component.\qed
\end{proof}

Our characterization of amenable graphs via Conditions~{\bfseries A},
{\bfseries B}, {\bfseries G} and~{\bfseries H} leads to an efficient
test for amenability of a given graph, that has the same time
complexity as CR. It is known (Cardon and Crochemore \cite{CardonC82};
see also \cite{BerkholzBG13}) that the stable partition of a given
graph $G$ can be computed in time $O((n+m)\log n)$. It is supposed
that $G$ is presented by its adjacency list.

\begin{restatable}{corollary}{recogn}\label{cor:recogn}
  The class of amenable graphs is recognizable in time $O((n+m)\log n)$,
where $n$ and $m$ denote the number of vertices and edges of the input graph.
\end{restatable}

\newcommand{\proofrecogn}{%
\begin{proof}
  Using known algorithms, we first compute the stable partition
  $\Pa_G=\{X_1,\dots,X_k\}$ of the input graph $G$. Let $C^*(G)$ be
  the version of the cell graph $C(G)$ where all empty edges are
  removed. We can compute the adjacency list of each vertex $X_i$
  of~$C^*(G)$ by traversing the adjacency list of an arbitrary vertex
  $u\in X_i$ and listing all cells $X_j$ that contain a vertex $v$
  adjacent to $u$. Simultaneously, we compute for each pair $(i,j)$
  such that $i=j$ or~$\{X_i,X_j\}$ is an edge of~$C^*(G)$ the
  number $d_{ij}$ of neighbors in $X_j$ of any vertex in
  $X_i$. Knowing the numbers $|X_i|$, $|X_j|$ and $d_{ij}$ allows us
  to determine whether all the subgraphs $G[X_i]$ and $G[X_i,X_j]$
  fulfill Conditions~{\bfseries A} and~{\bfseries B} of Lemma
  \ref{lem:necessaryAB}.

  To check Conditions~{\bfseries G} and~{\bfseries H} we use
  breadth-first search in the graph~$C^*(G)$ to find all anisotropic
  components $A$ of $C(G)$ and, simultaneously, to check that each
  component $A$ is a tree containing at most one heterogeneous
  cell. If we restart the search from an arbitrary cell in $A$ having
  minimum cardinality, we can also check for each forward edge of the
  resulting search tree whether the monotonicity property of
  Condition~{\bfseries G} is fulfilled.  \qed
\end{proof}
}{\proofrecogn}%

We conclude \ifcombined{this section }by considering logical aspects of
our result.  A \emph{counting quantifier} $\exists^m$ opens a sentence
saying that there are at least $m$ elements satisfying some property.
Immerman and Lander \cite{ImmermanL90} discovered an intimate
connection between color refinement and 2-variable first-order logic
with counting quantifiers.  This connection implies that amenability
of a graph is equivalent to its definability in this logic. Thus,
Corollary \ref{cor:recogn} asserts that the class of graphs definable
by a first-order sentence with counting quantifiers and occurrences of
just 2 variables is recognizable in polynomial time.
}

\ifcolref{\clearpage}{%
\section{Amenable graphs are compact}\label{sec:compact}

An $n\times n$ real matrix $X$ is \emph{doubly stochastic} if its
elements are nonnegative and all its rows and columns sum up to 1.
Doubly stochastic matrices are closed under products and convex
combinations. The set of all $n\times n$ doubly stochastic matrices
forms the \emph{Birkhoff polytope} $B_n\subset \reals^{n^2}$.
\emph{Permutation matrices} are exactly 0-1 doubly stochastic matrices.
By Birkhoff's Theorem, the $n!$ permutation matrices form precisely
the set of all extreme points of $B_n$. Equivalently, every doubly
stochastic matrix is a convex combination of permutation matrices.

Let $G$ and $H$ be graphs with vertex set $\{1,\ldots,n\}$.  An
isomorphism $\pi$ from $G$ to $H$ can be represented by the
permutation matrix $P_\pi=(p_{ij})$ such that $p_{ij}=1$ if and only if $\pi(i)=j$. Denote the
set of matrices $P_\pi$ for all isomorphisms $\pi$ by $\Iso(G,H)$, and
let $\Aut(G)=\Iso(G,G)$. 

Let $A$ and $B$ be the adjacency matrices of graphs $G$ and $H$
respectively. If the graphs are uncolored, a permutation matrix $X$ is in $\Iso(G,H)$ if
and only if $AX=XB$. 
For vertex-colored graphs, $X$ must additionally satisfy the condition
$X[u,v]=0$ for all pairs of differently colored $u$ and $v$, i.e.,
this matrix must be block-diagonal with respect to the color classes.
We say that (vertex-colored) graphs $G$ and $H$
are \emph{fractionally isomorphic} if $AX=XB$ for a doubly stochastic
matrix $X$, where $X[u,v]=0$ if $u$ and $v$ are of different
colors. The matrix $X$ is called a \emph{fractional isomorphism}.

Denote the set of all fractional isomorphisms from $G$ to $H$ by
$\ds{G,H}$ and note that it forms a polytope in $\reals^{n^2}$. 
The set of isomorphisms $\Iso(G,H)$ is
contained in $\extr{\ds{G,H}}$, where $\extr{Z}$ denotes the set of
all extreme points of a set $Z$. Indeed, $\Iso(G,H)$ is the set of
integral extreme points of $\ds{G,H}$. 

The set $\ds G=\ds{G,G}$ is the
polytope of \emph{fractional automorphisms} of $G$.
A graph $G$ is called \emph{compact} \cite{Tinhofer86} if $\ds G$ has
no other extreme points than $\Aut(G)$, i.e., $\extr{\ds G}=\Aut(G)$.
Compactness of $G$ can equivalently be defined by any of the following
two conditions:
\begin{itemize}[topsep=0mm]
\item The polytope $\ds G$ is integral;
\item Every fractional automorphism of $G$ is a convex combination of
  automorphisms of $G$, i.e., $\ds G=\langle\Aut(G)\rangle$, where
  $\langle Z\rangle$ denotes the convex hull of a set $Z$.
\end{itemize}

\begin{myexample}\label{ex:}
  Complete graphs are compact as a consequence of Birkhoff's theorem.
  The compactness of trees and cycles is established in \cite{Tinhofer86}. 
Matching graphs
  $mK_2$ are also compact. This is a particular instance of a much more
  general result by Tinhofer \cite{Tinhofer91}: If $G$ is compact,
  then $mG$ is compact for any~$m$. Tinhofer \cite{Tinhofer91} also
  observes that compact graphs are closed under complement.

  For a negative example, note that the graph $C_3+C_4$ is not
  compact. This follows from a general result in \cite{Tinhofer91}:
  All regular compact graphs must be vertex-transitive (and $C_3+C_4$
  is not).
\end{myexample}

 Tinhofer \cite{Tinhofer91} noted that, if
 $G$ is compact, then for any graph $H$, either all or
  none of the extreme points of the polytope $\ds{G,H}$ are integral.
As mentioned in the introduction, this yields a
linear-programming based polynomial-time algorithm to test if a
compact graph $G$ is isomorphic to any other given graph $H$.
The following
result shows that Tinhofer's approach works for all amenable graphs.


\begin{theorem}\label{thm:compact}
  All amenable graphs are compact. 
\end{theorem}

\ifproceedings{The proof of Theorem \ref{thm:compact} is in the
  appendix (see Section~\ref{sec:proofcompact}).}{We defer the proof
  to the next section.}  Theorem \ref{thm:compact} unifies and extends
several earlier results providing examples of compact graphs.  In
particular, it gives another proof of the fact that almost all graphs
are compact, which also follows from a result of Godsil
\cite[Corollary 1.6]{Godsil97}.  Indeed, while Babai, Erd{\"o}s, and
Selkow \cite{BabaiES80} proved that almost all graphs are discrete,
we already mentioned in Section \ref{sec:intro} that all
discrete graphs are amenable.

Furthermore, Theorem \ref{thm:compact} reproves Tinhofer's result that
trees are compact.\footnote{The proof of Theorem~\ref{thm:compact}
  uses only compactness of complete graphs, matching graphs, and the
  5-cycle.}  Since also forests are amenable \cite{arxiv}, we can
extend this result to forests.  This extension is not straightforward
as compact graphs are not closed under disjoint union; see Example
\ref{ex:}.  In \cite{Tinhofer89}, Tinhofer proves compactness for the
class of \emph{strong tree-cographs}, which includes forests only with
pairwise non-isomorphic connected components.

Compactness of unigraphs, which also follows from Theorem
\ref{thm:compact}, seems to be earlier never observed. Summarizing, we
note the following result.

\begin{corollary}
  Discrete graphs, forests, and unigraphs are compact.
\end{corollary}

\newcommand{\proofcompact}{%
\section{Proof of Theorem \protect\ref{thm:compact}}\label{sec:proofcompact}

We will use a known fact on the structure of fractional automorphisms.
For a partition $V_1,\dots,V_m$ of $\{1,\ldots,n\}$ let
$X_1,\dots,X_m$ be matrices, where the rows and columns of $X_i$ are
indexed by elements of $V_i$. Then we denote the block-diagonal matrix
with blocks $X_1,\dots,X_m$ by $\diag{X_1}{\dots}{X_m}$.
\begin{lemma}[Ramana et al.~\cite{RamanaSU94}]\label{lem:Ramana}
  Let $G$ be a (vertex-colored) graph on vertex set $\{1,\ldots,n\}$
  and assume that the elements $V_1,\dots,V_m$ of the stable partition
  $\Pa_G$ of $G$ are intervals of consecutive integers.  Then any
  fractional automorphism $X$ of $G$ has the form
  $X=\diag{X_1}{\dots}{X_m}$.
\end{lemma}
Note that the assumption of the lemma can be ensured for any graph by
appropriately renaming its vertices. An immediate consequence of
Lemma~\ref{lem:Ramana} is that a graph $G$ is compact if and only if it is
compact with respect to its stable coloring.

Given an amenable graph $G$ and a fractional automorphism $X$ of $G$,
we have to express $X$ as a convex combination of permutation matrices
in $\Aut(G)$. Our proof strategy consists in exploiting the structure
of amenable graphs as described by Theorem~\ref{thm:amenable}.  Given
an anisotropic component $A$ of the cell graph $C(G)$, we define the
\emph{anisotropic component $G_A$ of $G$} as the subgraph of $G$
induced by the union of all cells belonging to $A$.  Our overall idea
is to prove the claim separately for each anisotropic component $G_A$,
applying an inductive argument on the number of cells in $A$. A key
role will be played by the fact that, according to
Theorem~\ref{thm:amenable}, $A$ is a tree with at most one
heterogeneous vertex.

By Lemma \ref{lem:Ramana}, we can assume that $G$ is colored by the
stable coloring. We first consider the case when $G$ consists of a
single anisotropic component. By Theorem~\ref{thm:amenable}, the
corresponding cell graph $C(G)$ has at most one heterogeneous vertex,
and the anisotropic edges form a spanning tree of $C(G)$. Without loss
of generality, we can number the cells $V_1,\dots,V_m$ of $G$ so that
$V_1$ is the unique heterogeneous cell if it exists; otherwise $V_1$
is chosen among the cells of minimum cardinality.  Moreover, we can
suppose that, for each $i\le m$, the cells $V_1,\dots,V_i$ induce a
connected subgraph in the tree of anisotropic edges of~$C(G)$.

We will prove this case by induction on the number $m$ of cells.  In
the base case of $m=1$, our graph $G=G[V_1]$ is one of the graphs
listed in Condition~{\bfseries A} of Theorem~\ref{thm:amenable}. All
of them are known to be compact; see Example \ref{ex:}.  As induction
hypothesis, assume that the graph $H=G[V_1\cup\dots\cup V_{m-1}]$ is
compact.  For the induction step, we have to show that also
$G=G[V_1\cup\dots\cup V_m]$ is compact.

Denote $D=V_m$. Since $G$ has no more than one heterogeneous cell,
$G[D]$ is complete or empty.  It will be instructive to think of $D$
as a ``leaf'' cell having a unique anisotropic link to the remaining
part $H$ of $G$.  Let $C\in\{V_1,\dots,V_{m-1}\}$ be the unique cell
such that $\{C,D\}$ is an anisotropic edge of $C(G)$.  To be specific,
suppose that $G[C,D]\cong sK_{1,t}$.  If $G[C,D]$ is 
the bipartite complement of $sK_{1,t}$,
we can consider the complement of $G$ and use the
facts that the class of amenable graphs is closed under
complementation and that complementation does not change
fractional isomorphisms of the graph.
By the monotonicity property stated in
Condition~{\bfseries C} of Theorem~\ref{thm:amenable}, $|C|=s$ and
$|D|=st$. Let $ C=\{c_1,c_2,\dots,c_s\} $ and, for each $i$,
$N(c_i)\subset D$ be the neighborhood of $c_i$ in $G[C,D]$.  Thus,
$D=\bigcup_{i=1}^s N(c_i)$.

Let $X$ be a fractional automorphism of $G$. It is
convenient to break it up into three blocks
$X = \diag{X'}YZ$,
where $Y$ and $Z$ correspond to $C$ and $D$ respectively, and $X'$ is
the rest. By induction hypothesis we have the convex combination
\begin{equation}\label{eq:ds-C}
\ddiag{X'}Y = \sum_{\ddiag{P'}P\in\Aut(H)} \alpha_{P',P}\, \ddiag{P'}P,
\end{equation}
where $\ddiag{P'}P$ are permutation matrices corresponding to
automorphisms $\pi$ of the graph $H$, such that the permutation matrix
block $P$ denotes the action of $\pi$ on the color class $C$ and $P'$
the action on the remaining color classes of $H$.

We need to show that $X$ is a convex combination of automorphisms of
$G$. Let $A$ denote the adjacency matrix of $G$, and $A_{S,T}$ denote
the submatrix of $A$ 
row-indexed by $S\subset V(G)$ and column-indexed by $T\subset V(G)$.
Since $X$ is a fractional automorphism of $G$, we have
  $XA=AX$.  
Recall that $Y$ and $Z$ are blocks of $X$ corresponding to color
classes $C$ and $D$. Looking at the corner fragments of the matrices
$XA$ and $AX$, we get
\[
\left(
\begin{array}{cc}
  Y & \bigzero \\ 
  \bigzero &  Z \\
\end{array}
\right)
\left(
\begin{array}{cc}
  A_{C,C} & A_{C,D} \\ 
  A_{D,C} &  A_{D,D} \\
\end{array}
\right) = 
\left(
\begin{array}{cc}
  A_{C,C} & A_{C,D} \\ 
  A_{D,C} &  A_{D,D} \\
\end{array}
\right)
\left(
\begin{array}{cc}
  Y & \bigzero \\ 
  \bigzero &  Z \\
\end{array}
\right),
\]
which implies
\begin{eqnarray}\label{eq:ds-block1}
  Y A_{C,D} & = & A_{C,D}\, Z,\\ 
  A_{D,C}\, Y & = & Z\, A_{D,C}. \label{eq:ds-block2}
\end{eqnarray}

Consider $Z$ as an $st\times st$ matrix whose rows and columns are
indexed by the elements of sets $N(c_1), N(c_2),\dots, N(c_r)$ in that
order.  We can thus think of $Z$ as an $s\times s$ block matrix of
$t\times t$ matrix blocks $Z^{(k,\ell)}, 1\le k,\ell \le s$.  The next
claim is a consequence of Equations~\refeq{ds-block1} and
\refeq{ds-block2}.

\begin{restatable}{myclaim}{dsD}\label{cl:dsD}
  Each block $Z^{(k,\ell)}$ in $Z$ is of the form
  \begin{equation}\label{eq:ds-D}
    Z^{(k,\ell)}=y_{k,\ell}W^{(k,\ell)},
  \end{equation}
  where $y_{k,\ell}$ is the $(k,\ell)^{th}$ entry of $Y$, and
  $W^{(k,\ell)}$ is a doubly stochastic matrix.
\end{restatable}

\newcommand{\proofdsD}{%
\begin{proof}
  We first note from Equation~\refeq{ds-block1} that the $(k,j)^{th}$
  entry of the $s\times st$ matrix $Y A_{C,D} = A_{C,D} Z$ can be
  computed in two different ways. In the left hand side matrix, it is
  $y_{k,\ell}$ for each $j\in N(c_\ell)$. On the other hand, the right
  hand side matrix implies that the same $(k,j)^{th}$ entry is also
  the sum of the $j^{th}$ column of the $N(c_k)\times N(c_\ell)$ block
  $Z^{(k,\ell)}$ of the matrix $Z$.

  We conclude, for $1\le k,\ell\le s$, that each column in
  $Z^{(k,\ell)}$ adds up to~$y_{k,\ell}$.  By a similar argument,
  applied to Equation~\refeq{ds-block2} this time, it follows, for
  each $1\le k,\ell\le s$, that each \emph{row} of any block
  $Z^{(k,\ell)}$ of $Z$ adds up to $y_{k,\ell}$.

  We conclude that, if $y_{k,\ell}\ne0$, then the matrix
  $W^{(k,\ell)}=\frac1{y_{k,\ell}}Z^{(k,\ell)}$ is doubly stochastic.
  If $y_{k,\ell}=0$, then \refeq{ds-D} is true for any choice
  of~$W^{(k,\ell)}$.$~~\rule{1.0ex}{1.2ex}$
\end{proof}
}\ifproceedings{\proofdsD}{\proofdsD}%

For every $P=(p_{k\ell})$ appearing in an automorphism $\ddiag{P'}P$
of $H$ (see Equation~\refeq{ds-C}), we define the $st\times st$ doubly
stochastic matrix $W_P$ by its $t\times t$ blocks indexed by $1\le
k,\ell \le s$ as follows:
\begin{equation}\label{eq:W_P}
  W_P^{(k,\ell)}=
  \begin{cases}
    W^{(k,\ell)}&\text{if } p_{k\ell}=1,\\
    \bigzero&\text{if } p_{k\ell}=0.
  \end{cases}
\end{equation}
Equations \refeq{ds-C} and \refeq{ds-D} imply that
\begin{equation}\label{eq:ds-X}
  X=\diag{X'}YZ = \sum_{\ddiag{P'}P\in\Aut(H)} \alpha_{P',P}\, \diag{P'}P{W_P}.
\end{equation}
In order to see this, on the left hand side consider the
$(k,\ell)^{th}$ block $Z^{(k,\ell)}$ of $Z$.  On the right hand side,
note that the corresponding block in each $\diag{P'}P{W_P}$ is the
matrix $W^{(k,\ell)}$. Clearly, the overall coefficient for this block
equals the sum of $\alpha_{P',P}$ over all $P'$ and $P$ such that
$p_{k,\ell}=1$, which is precisely $y_{k,\ell}$ by
Equation~\refeq{ds-C}.

Since each $W^{(k,\ell)}$ is a doubly stochastic matrix, by Birkhoff's
theorem we can write it as a convex combination of $t\times t$
permutation matrices $Q_{j,k,\ell}$, whose rows are indexed by
elements of $N(c_k)$ and columns by elements of $N(c_\ell)$:
\[
W^{(k,\ell)}=\sum_{j=1}^{t!}\beta_{j,k,\ell}\, Q_{j,k,\ell}.
\]

Substituting the above expression in Equation~\refeq{W_P}, that defines the
doubly stochastic matrix $W_P$, we express $W_P$ as a convex
combination of permutation matrices
$W_P = \sum_{Q} \delta_{Q,P}\, Q$
where $Q$ runs over all $st\times st$ permutation matrices indexed by
the vertices in color class $D$. Notice here that $\delta_{Q,P}$ is
nonzero only for those permutation matrices $Q$ that have structure
similar to that described in Equation~\refeq{W_P}: The block
$Q^{(k,\ell)}$ is a null matrix if $p_{k\ell}=0$ and it is some
$t\times t$ permutation matrix if $p_{k\ell}=1$.  For each such $Q$,
the $(s+st)\times (s+st)$ permutation matrix $\ddiag PQ$ is an
automorphism of the subgraph $G[C,D]=sK_{1,t}$ (because $Q$ maps
$N(c_i)$ to $N(c_j)$ whenever $P$ maps $c_i$ to $c_j$).  Since
$P\in\Aut(G[C])$ and $D$ is a homogeneous set in $G$, we conclude
that, moreover, $\ddiag PQ$ is an automorphism of the subgraph
$G[C\cup D]$.

Now, if we plug the above expression for each $W_P$ in
Equation~\refeq{ds-X}, we will finally obtain the desired convex
combination
\[
X = \sum_{P',P,Q} \gamma_{P',P,Q}\, \diag{P'}PQ.
\]
It remains to argue that every $\diag{P'}PQ$ occurring in this sum is
an automorphism of $G$. Recall that a pair $P',P$ can appear here only
if $\ddiag{P'}P\in\Aut(H)$.  Moreover, if such a pair is extended to a
matrix $\diag{P'}PQ$, then $\ddiag PQ\in\Aut(G[C\cup D])$. Since
$G[B,D]$ is isotropic for every color class $B\ne D$ of $G$, we
conclude that $\diag{P'}PQ\in\Aut(G)$.  This completes the induction
step and finishes the case when $G$ has one anisotropic component.

\smallskip

Next, we consider the case when $C(G)$ has several anisotropic
components $T_1,\dots,\allowbreak T_k$, $k\ge2$. Let $G_1,\dots, G_k$,
where $G_i=G[\bigcup_{U\in V(T_i)}U]$, be the corresponding
anisotropic components of $G$.  By the proof of the previous case we
already know that $G_i$ is compact for each~$i$.

\begin{restatable}{myclaim}{autos}\label{cl:autos}
  The automorphism group $\Aut(G)$ of $G$ is the product of the
  automorphism groups $\Aut(G_i)$, $1\le i\le k$.
\end{restatable}

\newcommand{\proofautos}{%
\begin{proof}
  Recall that any automorphism of $G$ must map each color class of
  $G$, which is a cell of the underlying amenable graph $G'$, onto
  itself.  Thus, any automorphism $\pi$ of $G$ is of the form
  $(\pi_1,\dots,\pi_k)$, where $\pi_i$ is an automorphism of the
  subgraph $G_i$. Now, for any two subgraphs $G_i$ and $G_j$, we
  examine the edges between $V(G_i)$ and $V(G_j)$. For any color
  classes $U\subseteq V(G_i)$ and $U'\subseteq V(G_j)$, the edge
  $\{U,U'\}$ is isotropic because it is not contained in any
  anisotropic component of $C(G)$. Therefore, the bipartite graph
  $G[U,U']$ is either complete or empty. It follows that for any
  automorphisms $\pi_i$ of $G_i$, $1\le i\le k$, the permutation
  $\pi=(\pi_1,\dots,\pi_k)$ is an automorphism of the
  graph~$G$.$~~\rule{1.0ex}{1.2ex}$
\end{proof}
}\ifproceedings{\proofautos}{\proofautos}%

As follows from Lemma \ref{lem:Ramana}, any fractional automorphism
$X$ of $G$ is of the form
$X = \diag{X_1}{\dots}{X_k}$,
where $X_i$ is a fractional automorphism of $G_i$ for each $i$.  As
each $G_i$ is compact we can write each $X_i$ as a convex combination
$X_i = \sum_{\pi\in \Aut(G_i)}\alpha_{i,\pi}\, P_{\pi}$.
This implies
\begin{equation}\label{eq:eq3}
  \dddiagggg I\dots I{X_i}I\dots I = \sum_{\pi\in \Aut(G_i)}\alpha_{i,\pi}\,
  \dddiagggg I\dots I{P_{\pi}} I\dots I,
\end{equation}
where block diagonal matrices in the above expression have $X_i$ and
$P_{\pi}$ respectively in the $i^{th}$ block (indexed by elements of
$V(G_i)$) and identity matrices as the remaining blocks.

We now decompose the fractional automorphism $X$ as a matrix product
of fractional automorphisms of $G$
\[
  X = \diag{X_1}{\dots}{X_k} 
    = (\ddiagg{X_1}I\dots I)\cdot (\ddiagg I{X_2}\dots I)\cdot \dots
  \cdot (\ddiagg I\dots I{X_k}).
\]
Substituting for $\dddiagggg I\dots I{X_i}I\dots I$ from
Equation~\refeq{eq3} in the above expression and writing the product
of sums as a sum of products, we see that $X$ is a convex combination
of permutation matrices of the form $\diag{P_{\pi_1}}\dots{P_{\pi_k}}$
where $\pi_i\in\Aut(G_i)$ for each $i$. By Claim \ref{cl:autos}, all
the terms $\diag{P_{\pi_1}}\dots{P_{\pi_k}}$ correspond to
automorphisms of $G$. Hence, $G$ is compact, completing the proof of
Theorem \ref{thm:compact}.
}
\ifproceedings{\proofcompact}{\proofcompact}%

\section{A color-refinement based hierarchy of graphs}\label{s:hierarchy}

Let $u\in V(G)$ and $v\in V(H)$ be vertices of two graphs $G$ and
$H$. By \emph{individualization} of $u$ and $v$ we mean assigning the
same \emph{new color} to $u$ and $v$, which makes them distinguished
from the remaining vertices of $G$ and $H$. Tinhofer \cite{Tinhofer91}
proved that, if $G$ is compact, then the following polynomial-time
algorithm correctly decides if $G$ and $H$ are isomorphic.
\begin{enumerate}
\item Run CR on $G$ and $H$ until the coloring of
  $V(G)\cup V(H)$ stabilizes.
\item If the multisets of colors in $G$ and $H$ are different, then
  output ``non-isomorphic'' and stop. Otherwise,
\begin{enumerate}
\item if all color classes are singletons in $G$ and $H$, then if the
  map $u\mapsto v$ matching each vertex $u\in V(G)$ to the vertex $v\in V(H)$ of the
  same color is an isomorphism, output ``isomorphic'' and stop. Else
  output ``non-isomorphic'' and stop.
\item pick any color class with at least two vertices in both $G$ and
  $H$, select an arbitrary $u\in V(G)$ and $v\in V(H)$ in this color
  class and individualize them. Goto Step 1.
\end{enumerate}
\end{enumerate}

If $G$ and $H$ are any two non-isomorphic graphs, then Tinhofer's
algorithm will always output ``non-isomorphic''. However, it can fail
for isomorphic input graphs, in general. We call $G$ a \emph{Tinhofer
  graph} if the algorithm works correctly on $G$ and every $H$ for all
choices of vertices to be individualized. 
Thus, the result of \cite{Tinhofer91} can be stated as the inclusion
$\compact\subseteq\tinhofer$.

If $G$ is a Tinhofer graph,
then the above algorithm can be used to even find a canonical labeling of $G$. 
Using Theorem \ref{thm:compact}, we state the following fact.

\begin{corollary}
  Amenable and, more generally, compact graphs admit canonical
  labeling in polynomial time.
\end{corollary}

\ifcombined{}{%
A partition $\Pa$ of the vertex set of a graph $G$ is \emph{equitable}
if, for any two elements $X$ and $Y$ of $\Pa$, every vertex $x\in X$
has the same number of neighbors in $Y$. Note that the stable
partition $\Pa_G$ produced by CR on the input $G$ is equitable.}

Let $A$ be a subgroup of the automorphism group $\auto G$ of a graph
$G$.  Then the partition of $V(G)$ into the $A$-orbits is called an
\emph{orbit partition} of $G$.  Any orbit partition of $G$ is
equitable, but the converse is not true, in general.  However, Godsil
\cite[Corollary 1.3]{Godsil97} has shown that the converse holds for
compact graphs. We define
\emph{Godsil graphs} as the graphs for which the two notions
of an equitable and an orbit partition coincide.  
Thus, the result of \cite{Godsil97} can be stated as the inclusion
$\compact\subseteq\godsil$.
Now, the inclusion $\compact\subseteq\tinhofer$
can easily be strengthened as follows.

\begin{restatable}{lemma}{GodTin}\label{lem:God-Tin}
  Any Godsil graph is a Tinhofer graph.
\end{restatable}

\newcommand{\proofGodTin}{%
\begin{proof}
  Assume that $G$ is a Godsil graph. It suffices to show that
  Tinhofer's algorithm is correct whenever $G$ and $H$ are isomorphic.
  Let $\phi$ be an isomorphism from $G$ to $H$. We will prove that,
  after the $i$-th refinement step made by the algorithm, there exists
  an isomorphism $\phi_i$ from $G$ to $H$ that preserves colors of the
  vertices. If this is true for each $i$, the algorithm terminates
  only if the discrete partition (i.e., the finest partition into
  singletons) is reached. Suppose that this happens in the $k$-th
  step. Then $\phi_k$ ensures that the algorithm decides isomorphism.

  We prove the claim by induction on $i$.  At the beginning,
  $\phi_1=\phi$. Assume that an isomorphism $\phi_i$ exists and the
  partition is still not discrete.  Suppose that now the algorithm
  individualizes $u\in V(G)$ and $v\in V(H)$.  If $v=\phi_i(u)$, then
  $\phi_{i+1}=\phi_i$. Otherwise, consider the vertices $u$ and
  $\phi^{-1}_i(v)$, which are in the same monochromatic class of $G$.
  Note that the partition of $G$ produced in each refinement step is
  equitable. Since $G$ is Godsil, there is an automorphism $\alpha$
  preserving the partition such that $\alpha(u)=\phi^{-1}_i(v)$. We
  can, therefore, take $\phi_{i+1}=\phi_i\circ\alpha$.\qed
\end{proof}
}\ifproceedings{}{\proofGodTin}%

The orbit partition of $G$ with respect to $\auto G$ is always a
refinement of the stable partition $\Pa_G$ of $G$. We call $G$
\emph{refinable} if $\Pa_G$ is the orbit
partition of $\auto G$.
It is easy to show the following.

\begin{restatable}{lemma}{Tinref}\label{lem:Tin-ref}
  Any Tinhofer graph is refinable.
\end{restatable}

\newcommand{\proofTinref}{%
\begin{proof}
  Suppose that $G$ is not refinable. Then $G$ has vertices $u$ and $v$
  that are in different orbits but not separated by the stable
  partition $\Pa_G$. This means that individualization of $u$ and
  $v$ in isomorphic copies $G'$ and $G''$ of $G$ gives non-isomorphic
  results. Therefore, if Tinhofer's algorithm is run on $G'$ and $G''$
  and individualizes $u$ and $v$, it eventually decides that $G'$ and
  $G''$ are non-isomorphic.\qed
\end{proof}
}\ifproceedings{}{\proofTinref}%

Summarizing Theorem \ref{thm:compact}, Lemmas \ref{lem:God-Tin} and
\ref{lem:Tin-ref}, and \cite[Corollary 1.3]{Godsil97}, we 
state the following hierarchy result.

\begin{theorem}\label{thm:hierarchy}
  The classes of graphs under consideration form the inclusion chain
  \newcounter{old}\setcounter{old}{\theequation}\setcounter{equation}{0}
  \begin{equation}
    \discrete\subset \amenable\subset\compact\subset\godsil\subset\tinhofer\subset\refinable.
  \end{equation}\setcounter{equation}{\theold}
  Moreover, all of the inclusions are strict.
\end{theorem}

It is worth of noting that the hierarchy \refeq{hier-classes}
collapses to \discrete if we restrict ourselves to only rigid graphs,
i.e., graphs with trivial automorphism group.
 
\newcommand{\proofSepRefTin}{%
  Consider the gadget $\CFI(P_1,P_2,P_3)$ depicted in
  Figure~\ref{fig:gadgets}, with two input pairs $P_1$ and $P_2$ and
  one output pair $P_3$. This gadget \cite{CFI92} has the property
  that any automorphism of it must flip an even number of the three
  pairs $P_1$, $P_2$, and $P_3$. We can combine $\CFI(P_1,P_2,P_3)$
  with a second gadget $\CFI(P_1,P_2,P_4)$ with the same input pairs
  and a fresh output pair $P_4$. We assume that the four pairs $P_1$,
  $P_2$, $P_3$, and $P_4$ and the intermediate sets of four connecting
  vertices, are all different color classes.

  This defines a refinable graph $G$, also depicted in
  Figure~\ref{fig:gadgets}, with four color classes $P_1$, $P_2$,
  $P_3$, and $P_4$ of size 2, and two color classes $F$ and $F'$ of
  size 4 corresponding to the orbit partition of $G$. The graph $G$
  has the property that any automorphism of it must flip either both
  pairs $P_3$ and $P_3$ or none of them. Now, if we run the Tinhofer
  procedure on two identical copies $G'$ and $G''$ of $G$, it might
  individualize color class $P_3$ in the first round and color class
  $P_4$ in the second round in such a way that the resulting graphs
  are not isomorphic, since the partial isomorphism flips exactly one
  of the two color classes. Note that the vertex colors of $G$ can be
  removed if we connect the four vertices in $F$ by six edges and the
  two vertices in $P_1$ by one edge.
}

\newcommand{\proofstrict}{%
  The following separating examples prove that all inclusions are
  strict.
\begin{description}
\item{Separation of \discrete and \amenable:} For $n\ge2$, the
  complete graph $K_n$ is amenable but not discrete.

\item{Separation of \amenable and \compact:} For $n\ge6$, the cycles
  $C_n$ are not amenable because they are indistinguishable from a
  pair $C_{n_1}+C_{n_2}$ of disjoint cycles on $n_1+n_2=n$
  vertices. On the other hand, cycles are known to be compact graphs
  \cite[Theorem~2]{Tinhofer86}.

\item{Separation of \compact and \godsil:} These classes are separated
  by the well-known Petersen graph. Evdokimov, Karpinski, and
  Ponomarenko \cite[Corollary 5.4]{EvdokimovKP99} prove that the
  Petersen graph is not compact. They explicitly give a fractional
  automorphism of the Petersen graph which cannot be written as a
  convex combination of its automorphisms.  It remains to show that
  the Petersen graph belongs to the class \godsil.

  This problem is solvable by modern computer algebra tools; see
  \cite{Ziv-Av} where equitable and orbit partitions are counted for
  various strongly regular graphs, including the Petersen graph.  We
  give a non-computer-assisted proof in Section \ref{ssec:petg}.

\item{Separation of \godsil and \tinhofer:} These classes are
  separated by the Johnson graphs $J(n,2)$ for $n\ge7$. The
  \emph{Johnson graph} $J(n,k)$ has the $k$-element subsets of
  $[n]=\{1,\ldots,n\}$ as vertices; any two of them are adjacent if
  their intersection consists of $k-1$ elements. Note that
  $J(n,1)=K_n$.  Furthermore, the graph $J(n,2)$ is the line graph of
  $K_n$: it has all 2-element subsets of $[n]$ as vertices and any two
  of them are adjacent if their intersection is non-empty. It is
  noticed in \cite{ChanG97} that $J(n,2)$ is not Godsil for $n\ge7$.
  For establishing the separation, we show that $J(n,2)$ is indeed
  Tinhofer.  The proof is given in Section \ref{ssec:john-tin}.

\item {Separation of \tinhofer and \refinable:} \ifproceedings{A graph
    separating these classes is exhibited in
    Section~\ref{ssec:tin-ref}.
}{\proofSepRefTin}
\end{description}
}
\ifproceedings{%
The strictness part of Theorem \ref{thm:hierarchy} is proved in
Appendix~\ref{sa:hier}.
}{\proofstrict}

Finally, we show that testing membership in each of the graph classes
in the hierarchy \refeq{hier-classes} is \p-hard.

\begin{restatable}{theorem}{phard}\label{thm:phard}
  The recognition problem of each of the classes in the hierarchy
  \refeq{hier-classes} is \p-hard under uniform AC$^0$ many-one
  reductions.
\end{restatable}

\newcommand{\figphard}{%
\begin{figure}[t]
  \centering
    \begin{tikzpicture}[baseline=10pt,xscale=.55,yscale=.8]
      \tikzstyle{every edge}=[draw,color=black] \tikzstyle{every
        node}=[fill,inner sep=2pt,circle] \draw[line width=10pt,line
      cap=round,color=gray!55] (.9,0) to (2.1,0)
      node[fill=none,color=black,left=10pt,below] {$P_i$}; \draw[line
      width=10pt,line cap=round,color=gray!55] (3.9,0) to (5.1,0)
      node[fill=none,color=black,left=10pt,below] {$P_j$}; \draw[line
      width=10pt,line cap=round,color=gray!55] (2.4,3) to (3.6,3)
      node[fill=none,color=black,left=10pt,above] {$P_k$}; \draw[line
      width=10pt,line cap=round,color=gray!55] (1.4,1.5) to (4.6,1.5)
      node[fill=none,color=black,right=5pt,above=-1.5mm] {$F_k$}; \path (1,0) node
      (a0) {} (2,0) node (a1) {} (4,0) node (b0) {} (5,0) node (b1) {}
      (1.5,1.5) node (e00) {} (2.5,1.5) node (e01) {} (3.5,1.5) node (e10)
      {} (4.5,1.5) node (e11) {} (2.5,3) node (c0) {} (3.5,3) node (c1)
      {}; \path[-,thick] (e00) edge (a0) edge (b0) edge (c0);
      \path[-,thick] (e01) edge (a0) edge (b1) edge (c1);
      \path[-,thick] (e10) edge (a1) edge (b0) edge (c1);
      \path[-,thick] (e11) edge (a1) edge (b1) edge (c0);
      \node[fill=none] at (3,-1.5) {$\CFI(P_i,P_j,P_k)$}; 
    \end{tikzpicture}\hfil
    \begin{tikzpicture}[baseline=10pt,xscale=.55,yscale=.8]
      \tikzstyle{every edge}=[draw,color=black] \tikzstyle{every
        node}=[fill,inner sep=2pt,circle] \draw[line width=10pt,line
      cap=round,color=gray!55] (.9,1) to (2.1,1)
      node[fill=none,color=black,left=20pt,below=-1mm] {$P'_i$}; \draw[line
      width=10pt,line cap=round,color=gray!55] (3.9,1) to (5.1,1)
      node[fill=none,color=black,right=2pt,below=-1mm] {$P''_i$};
      \draw[line width=10pt,line cap=round,color=gray!55] (2.4,0) to
      (3.6,0) node[fill=none,color=black,left=10pt,below] {$P_i$};
      \draw[line width=10pt,line cap=round,color=gray!55] (2.4,3) to
      (3.6,3) node[fill=none,color=black,left=10pt,above] {$P_k$};
      \draw[line width=10pt,line cap=round,color=gray!55] (1.4,2) to
      (4.6,2) node[fill=none,color=black,right=5pt,above=-1.5mm] {$F_{ik}$}; \path
      (1,1) node (a0) {} (2,1) node (a1) {} (4,1) node (b0) {} (5,1)
      node (b1) {} (1.5,2) node (e00) {} (2.5,2) node (e01) {} (3.5,2)
      node (e10) {} (4.5,2) node (e11) {} (2.5,3) node (c0) {} (3.5,3)
      node (c1) {} (2.5,0) node (d0) {} (3.5,0) node (d1) {};
      \path[-,thick] (d0) edge (a0) edge (b0); \path[-,thick] (d1)
      edge (a1) edge (b1); \path[-,thick] (e00) edge (a0) edge (b0)
      edge (c0); \path[-,thick] (e01) edge (a0) edge (b1) edge (c1);
      \path[-,thick] (e10) edge (a1) edge (b0) edge (c1);
      \path[-,thick] (e11) edge (a1) edge (b1) edge (c0);
      \node[fill=none] at (3,-1.5) {$\IMP(P_i,P_k)$};
    \end{tikzpicture}\hfil
    \begin{tikzpicture}[baseline=10pt,xscale=.55,yscale=.8]
      \tikzstyle{every edge}=[draw,color=black] \tikzstyle{every
        node}=[fill,inner sep=2pt,circle] \draw[line width=10pt,line
      cap=round,color=gray!55] (2.4,0) to (3.6,0)
      node[fill=none,color=black,left=10pt,below] {$P_1$}; \draw[line
      width=10pt,line cap=round,color=gray!55] (6.4,0) to (7.6,0)
      node[fill=none,color=black,left=10pt,below] {$P_2$}; \draw[line
      width=10pt,line cap=round,color=gray!55] (2.4,3) to (3.6,3)
      node[fill=none,color=black,left=10pt,above] {$P_3$}; \draw[line
      width=10pt,line cap=round,color=gray!55] (6.4,3) to (7.6,3)
      node[fill=none,color=black,left=10pt,above] {$P_4$}; \draw[line
      width=10pt,line cap=round,color=gray!55] (5.4,1.5) to (8.6,1.5)
      node[fill=none,color=black,right=2pt,above=-1mm] {$F'$};
      \draw[line width=10pt,line cap=round,color=gray!55] (1.4,1.5) to
      (4.6,1.5) node[fill=none,color=black,right=2pt,above=-.6mm]
      {$F$}; \path (2.5,0) node (a0) {} (3.5,0) node (a1) {} (6.5,0)
      node (b0) {} (7.5,0) node (b1) {} (1.5,1.5) node (e00) {}
      (2.5,1.5) node (e01) {} (3.5,1.5) node (e10) {} (4.5,1.5) node
      (e11) {} (5.5,1.5) node (v00) {} (6.5,1.5) node (v01) {}
      (7.5,1.5) node (v10) {} (8.5,1.5) node (v11) {} (2.5,3) node
      (c0) {} (3.5,3) node (c1) {} (6.5,3) node (w0) {} (7.5,3) node
      (w1) {}; \path[-,thick] (e00) edge (a0) edge (b0) edge (c0);
      \path[-,thick] (e01) edge (a0) edge (b1) edge (c1);
      \path[-,thick] (e10) edge (a1) edge (b0) edge (c1);
      \path[-,thick] (e11) edge (a1) edge (b1) edge (c0);
      \path[-,thick] (v00) edge (b0) edge (a0) edge (w0);
      \path[-,thick] (v01) edge (b0) edge (a1) edge (w1);
      \path[-,thick] (v10) edge (b1) edge (a0) edge (w1);
      \path[-,thick] (v11) edge (b1) edge (a1) edge (w0);
      \node[fill=none] at (5,-1.5) {$G$};
    \end{tikzpicture}\vspace{-8mm}
  \caption{\label{fig:gadgets}The $\CFI(P_i,P_j,P_k)$- and
    $\IMP(P_i,P_k)$-gadgets and a graph $G$ separating $\refinable$ from
    $\tinhofer$}
\end{figure}
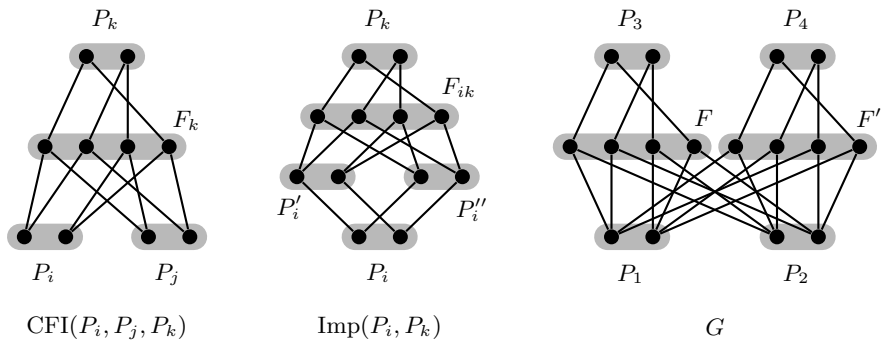}
\newcommand{\proofphard}{%
\begin{proof}
We sketch the proof. Given a monotone boolean circuit $C$ with and-
  and or-gates and constant input gates we construct a graph $G$ as
  follows:
\begin{itemize}[topsep=.5mm]
\item For each gate $g_k$ of $C$, $G$ contains a pair $P_k=\{a_k, b_k\}$ of
  vertices.
\item If $g_k$ is a constant input gate with value 1, then $a_k$ and
  $b_k$ get different colors (i.e., they form singleton color
  classes); otherwise $a_k$ and $b_k$ both get the same color (i.e.,
  they form a color class of size 2).
\item For each and-gate $g_k$ with input gates $g_i$ and $g_j$, $G$
  additionally contains a color class $F_k$ of size 4 that together
  with the two input pairs $P_i$ and $P_j$ as well as the output pair
  $P_k$ forms a $\CFI(P_i, P_j, P_k)$-gadget; see
  Figure~\ref{fig:gadgets}.
\item For each or-gate $g_k$ with input gates $g_i$ and $g_j$, $G$
  additionally contains two color classes $F_{ik}$ and $F_{jk}$ of
  size four, and four color classes $P'_i$, $P''_i$, $P'_j$, $P''_j$
  of size 2. The color classes $P'_i$, $P''_i$, $P_k$ and $F_{ik}$
  form a $\CFI(P'_i, P''_i, P_k)$-gadget and each of the pairs $P'_i$
  and $P''_i$ is linked to $P_i$ by two parallel edges. Henceforth, we
  denote this gadget by $\IMP(P_i,P_k)$; see
  Figure~\ref{fig:gadgets}. Likewise, the color classes $P'_j$,
  $P''_j$ and $F_{jk}$ are used to form an $\IMP(P_j,P_k)$-gadget.
\end{itemize}

A straightforward induction on the height of the and- and or-gates in
$C$ shows that CR on input $G$ refines a color class $P_k$ if and only
if the corresponding gate $g_k$ outputs value 1. This follows from the
following observations.
\begin{itemize}[itemsep=.5mm]
\item If $g_k$ is an and-gate with input gates $g_i$ and $g_j$, then
  the vertices in $P_k$ get different $C^{r+2}$ colors if and only if
  the vertices in $P_i$ as well as the vertices in $P_j$ have
  different $C^{r}$ colors.
\item If $g_k$ is an or-gate with input gates $g_i$ and $g_j$, then
  the vertices in $P_k$ get different $C^{r+3}$ colors if and only if
  either the vertices in $P_i$ or the vertices in $P_j$ have different
  $C^{r}$ colors.
\end{itemize}
Now let $G'$ be the graph that results from $G$ by connecting the
vertex pair $P_l$ corresponding to the output gate $g_l$ by two
parallel edges with each pair $P_k$ corresponding to a constant 0
input gate $g_k$. Then $C$ evaluates to 1 if and only if $G'$ is
discrete (i.e., CR on input $G'$ individualizes all vertices of $G'$).

Moreover, if we connect the output pair $P_l$ via an additional
$\IMP(P_l,P_{l+1})$-gadget to a new vertex pair $P_{l+1}$, then the
resulting graph $G''$ is not even refinable if $C$ evaluates to 0. The
reason is that no automorphism of $G''$ flips the pair $P_{l+1}$, but
CR only refines the color class $P_{l+1}$ if $C$ evaluates to 1.

Hence, the mapping $C\mapsto G''$ simultaneously reduces MCVP to each
of the graph classes in the hierarchy \refeq{hier-classes}.  \qed
\end{proof}

We observe that the graph $G''$ used in the proof of the hardness
results can be easily replaced by an uncolored graph. In fact, the
vertex colors can be substituted by suitable graph gadgets in such a
way that the automorphism group as well as the stable partition remain
essentially unchanged (up to the addition of several singleton
cells). Hence, the hardness results are also valid for the restricted
versions of the classes in the hierarchy \refeq{hier-classes} where we
consider only uncolored graphs.}
\ifproceedings{}{%
\figphard
\proofphard
}




}


\newpage
\appendix

\ifproceedings{%
\section{Proofs of Lemmas \protect\ref{lem:God-Tin} and \protect\ref{lem:Tin-ref}}

\GodTin*
\proofGodTin

\Tinref*
\proofTinref
}{}%


\section{Missing parts of the proof of Theorem
  \protect\ref{thm:hierarchy}}\label{sa:hier}

\ifproceedings{%
We recall the chain of inclusions in Theorem \ref{thm:hierarchy}:
\[
  \discrete\subset  \amenable\subset\compact\subset\godsil\subset\tinhofer\subset\refinable.
\]
\proofstrict
}{}%

\subsection{The Petersen Graph is \godsil}\label{ssec:petg}

It is well-known that the Petersen graph, denoted by $P$, is
isomorphic to the Kneser graph $K(5,2)$. The \emph{Kneser graph}
$K(n,k)$ has the $k$-element subsets of $[n]=\{1,\ldots,n\}$ as
vertices and any two of them are adjacent if they are disjoint.  An
important fact about $K(5,2)$ is that its automorphism group is
isomorphic to the symmetric group $S_5$ acting on the set
$\{1,\ldots,5\}$.  In fact, any automorphism of $K(5,2)$ can be
realized by extending the action of a permutation $\pi \in S_5$ to the
vertex set of $K(5,2)$~\cite{Whit32}.

First, we state some useful facts about the Petersen graph.
\begin{proposition}\label{thm:pet}
  The Petersen graph has the following properties:
  \begin{enumerate}[leftmargin=9mm,label=(\roman*),topsep=1mm,itemsep=.5mm]
  \item There are no cycles of length 3, 4 and 7.
  \item There are no independent sets of size greater than $4$.
  \item Any two adjacent vertices have no common neighbors and any two
    non-adjacent vertices have a unique common neighbor.
  \end{enumerate}
\end{proposition}

We will need some definitions regarding partitions of the vertex set
of a graph $G=(V,E)$.  Given a partition $\Sigma = \{S_1,\dots,S_k\}$
of $V$, we refer to the sets $S_1,\dots,S_k$ as the \emph{cells} of
$\Sigma$.  If the size of a cell is $k$, we call it a $k$-\emph{cell}.
Two cells $S$ and $S'$ are said to be \emph{compatible} if the induced
bipartite graph $P[S,S']$ is biregular (it can be empty).  Otherwise,
we say they are \emph{incompatible}.  Recall that any cell $S$ of an
equitable partition induces a regular graph $G[S]$.  Moreover, in that
case, any two cells $S,S'$ are compatible and the number of edges in
the biregular graph $G[S,S']$ is a common multiple of $|S|$ and
$|S'|$.

Now we are ready to prove the following theorem.

\begin{theorem}
  The Petersen graph $P$ is \godsil.
\end{theorem}
\begin{proof}
  To prove the theorem, we will enumerate all equitable partitions of
  $P$. For each such partition $\Sigma$, we describe a subgroup of
  $\auto P$ such that its orbit partition coincides with $\Sigma$. We
  represent the vertices of $P$ by the two-element subsets of the set
  $\Omega=\{a,b,c,d,e\}$, where two vertices are adjacent if they are
  disjoint. This representation allows us to describe any subgroup of
  $\auto P$ as a subgroup of the permutation group $S_{\Omega}$ on
  $\Omega$.


  The two trivial partitions of $V(P)$ into one set and into ten
  singleton sets are clearly orbit partitions, since the Petersen
  graph is vertex-transitive.  For our case analysis, we classify the
  remaining non-trivial equitable partitions of $P$ by the minimum
  size $\delta$ of the cells in the partition.  Clearly, $\delta \leq
  5$.  In the following claims we show for each $k\in \{1,2,3,4,5\}$
  that any equitable partition of $P$ with $\delta=k$ is an orbit
  partition of $P$.

  \begin{myclaim}\label{cl:delta3}
    $P$ does not have any equitable partition with $\delta=3$.
  \end{myclaim}

  \begin{proof}
    Suppose that there is an equitable partition $\Sigma$ with
    $\delta=3$ and let $S$ be a 3-cell in it. Then $\Sigma$ either has
    the form $\Sigma=\{S,T\}$, where $|T|=7$, or the form
    $\Sigma=\{S,U,V\}$ where $|U|=3$ and $|V|=4$.  The first case is
    ruled out since $P[T]$ can never be regular ($P$ has neither
    independent sets of size $7$ nor cycles of size $7$). Suppose the
    second case is possible. Then $P[S]$ and $P[U]$ must be empty
    (since $P$ has no triangles). Furthermore, the bipartite graphs
    $P[S,V]$ and $P[U,V]$ must be both biregular. The graph $P[S,V]$
    (likewise, $P[U,V]$) is empty or it has $12$ edges. It is not
    possible that $P[S,V]$ has $12$ edges because then $P[V]$ has only
    $3$ edges and cannot be regular. If both $P[S,V]$ and $P[U,V]$ are
    empty then $V$ is disconnected from the rest of the graph, which
    is a contradiction.$~~\rule{1.0ex}{1.2ex}$
  \end{proof}

  \begin{myclaim}\label{cl:delta4}
    All equitable partitions of $P$ with $\delta=4$ are orbit partitions.
  \end{myclaim}

  \begin{proof}
    We first show that any equitable partition $\Sigma$ with
    $\delta=4$ has one 4-cell $S$ and one 6-cell $T$, where $P[S]$ is
    empty and $P[T]$ is a 3-matching (a matching with 3
    edges). Clearly, $\Sigma$ must be of the form $\{S,T\}$, where
    $|S|=4$ and $|T|=6$. Moreover, $P[S]$ must be empty (0-regular) or
    2-matching (1-regular) since it cannot be a 4-cycle
    (2-regular). In fact, the case of 2-matching can also be ruled out
    by counting the number of edges as follows. For $S$ and $T$ to be
    compatible, there must be $12$ edges in the graph $P[S,T]$. Then
    there is exactly one edge left in the induced graph $P[T]$ which
    is impossible. Therefore, $P[S]$ must be empty. This also implies
    that the graph $P[S,T]$ has $4 \times 3 =12$ edges. Hence, $P[T]$
    must be a 3-matching.

    Now observe that any independent-set $S$ of size 4 in $P$ must be
    of the kind $S=\{ab,ac,ad,ae\}$ (up to automorphisms), implying
    that $T=\{bc,bd,be,cd,ce,de\}$.  The partition $\{S,T\}$ can be
    easily verified to be equitable and that it is the orbit partition
    of the subgroup $S_{\{b,c,d,e\}}$.$~~\rule{1.0ex}{1.2ex}$
  \end{proof}

  \begin{myclaim}\label{cl:delta5}
    All equitable partitions of $P$ with $\delta=5$ are orbit partitions.
  \end{myclaim}

  \begin{proof}
    In this case $\Sigma$ must have the form $\Sigma=\{S,T\}$ where
    $|S|=|T|=5$. Moreover, since $P$ does not have independent sets of
    size $5$, $P[S]$ and $P[T]$ must be 5-cycles. Clearly, such
    partitions exist, and any such partition has a matching between
    sets $S$ and $T$.

    It remains to show that $\Sigma=\{S,T\}$ is indeed an orbit
    partition of some subgroup of $\auto P$. Denote the 5-cycle in $S$
    by 1-2-3-4-5. Let $1'$ be the matching partner of $1$ in $T$ and
    so on. Now, $1'$ and $2'$ cannot be adjacent, else there is a
    4-cycle in $P$. The unique common neighbor of $1'$ and $2'$ must
    be $4'$, otherwise it is easy to verify that we will have a
    4-cycle in $P$. The partners $3'$ and $5'$ can also be uniquely
    determined in $T$.  The permutation $\pi = (12345)(1'2'3'4'5')$
    can be verified to be an automorphism of $P$ and the orbits of the
    subgroup generated by $\pi$ are precisely
    $\{S,T\}$.$~~\rule{1.0ex}{1.2ex}$
  \end{proof}

  \begin{myclaim}\label{cl:delta2}
    All equitable partitions of $P$ with $\delta=2$ are orbit partitions.
  \end{myclaim}

  \begin{proof}
    Let $\Sigma$ be an equitable partition of $P$ with $\delta=2$ and
    let $S=\{u,v\}$ be a 2-cell in it. We first show that $uv$ must
    be an edge.  This holds because any two non-adjacent vertices have
    a unique common neighbor $x$. The cell containing $x$ can only be
    a singleton set, which contradicts $\delta=2$.

    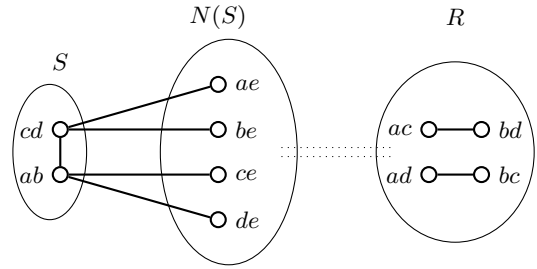
\begin{figure}
      \centering
      \begin{tikzpicture} [lineDecorate/.style={-,thick},xscale=.7,yscale=.6,%
        nodeDecorate/.style={shape=circle,inner sep=2pt,draw,thick}]
        \foreach \nodename/\x/\y/\direction/\navigate in {
          ab/0/0/left/west, cd/0/1/left/west, ae/3/2/right/east,
          be/3/1/right/east,ce/3/0/right/east,de/3/-1/right/east,
          ac/7/1/left/west, bd/8/1/right/east, ad/7/0/left/west,
          bc/8/0/right/east} { \node (\nodename) at (\x,\y)
          [nodeDecorate] {}; \node [\direction] at
          (\nodename.\navigate) {\footnotesize$\nodename$}; }
        \path \foreach \startnode/\endnode in
        {ab/cd,ac/bd,ad/bc,cd/ae,cd/be,ab/ce,ab/de} { (\startnode)
          edge[lineDecorate] node {} (\endnode) }; \draw (-.2,0.5)
        ellipse (0.7cm and 1.5cm); \node at (0,2.5) {$S$}; \draw
        (3.2,0.5) ellipse (1.3cm and 2.5cm); \node at (3,3.5) {$N(S)$};
        \draw (7.5,0.5) ellipse (1.5cm and 2cm); \node at (7.5,3.5)
        {$R$}; \draw[dotted] (4.2,.6) -- (6.3,.6); \draw[dotted]
        (4.2,.4) -- (6.3,.4);
      \end{tikzpicture}
      \caption{\label{fig:delta2}The case $\delta=2$.}
    \end{figure}

    Next we show that the neighborhood $N(S)=\bigcup_{x\in
      S}N(x)\setminus S$ of $S$ is also a cell of $\Sigma$ (see
    Figure~\ref{fig:delta2}).  Since $uv$ is an edge, there are no
    common neighbors of $u$ and $v$.  Therefore, $|N(S)|=4$.
    Moreover, $N(S)$ is an independent set since any edge among
    vertices in $N(S)$ can be used to construct a cycle of length 3 or
    4 passing through the edge $uv$.  This is not possible by
    Proposition~\ref{thm:pet}. Now let $R = V(P) \backslash (S \cup
    N(S))$ be the set of the four remaining vertices as shown in
    Figure~\ref{fig:delta2}.  Observe that no cell can contain
    vertices from both $N(S)$ and $R$, since then it would be
    incompatible with $S$.  Since $N(S)$ is an independent set, there
    cannot be a 2-cell inside $N(S)$. Clearly, there cannot be
    $1$-cells and hence 3-cells inside $N(S)$.  Therefore, $N(S)$ must
    indeed be a cell.

    By accounting for edges of $S$ and $N(S)$, it is easy to verify
    that $R$ has exactly two edges, and hence $P[R]$ must be a
    2-matching.  Since $\delta=2$, $R$ does not contain any $1$-cell
    and hence, any $3$-cells.  This leaves us with only two cases.
  \begin{description}[leftmargin=1mm,labelsep=1.25mm,itemindent=0mm,topsep=1.5mm,itemsep=1.5mm]
  \item[\mdseries Case 1:] $R$ is a cell.  We characterize all such partitions
    by naming a typical case.  W.l.o.g, let $S=\{ab,cd\}$ since $S$ is
    an edge.  Then $N(S)$ must be $\{ae,be,ce,de\}$ and $R$ must be
    $\{ac,ad,bc,bd\}$.  The partition
    $\{ab,cd\},\{ae,be,$ $ce,de\},\{ac,ad,bc,bd\}$ can be easily verified
    to be equitable. Moreover, it is easy to check that it is the
    orbit partition of the subgroup of all permutations in
    $S_{\Omega}$ which preserve the $\Omega$-partition
    $\{ab\},\{cd\},\{e\}$.  This is also the subgroup generated by the
    automorphisms $(ab),(cd),(ac)(bd)$.
  \item[\mdseries Case 2:] $\Sigma$ partitions $R$ in two sets $A$ and $B$ where
    $|A|=|B|=2$. Since each $2$-cell has to be an edge (see above),
    the sets $A$ and $B$ must be $\{ac,bd\}$ and $\{bc,ad\}$.  The
    partition $\{ab,cd\},\{ae,be,$ $ce,de\},\{ac,bd\},\{ad,bc\}$ can be
    easily verified to be equitable.  Moreover, it is easy to check
    that it is the orbit partition of the subgroup of all
    permutations in $S_{\Omega}$ which preserve the $\Omega$-partition
    $\{ab\},\{cd\},\{e\}$ and additionally, stabilize the sets
    $\{ac,bd\}$ and $\{ad,bc\}$. This is also the subgroup generated
    by the automorphisms $(ac)(bd), (ad)(bc), (ab)(cd)$.$~~\rule{1.0ex}{1.2ex}$
  \end{description}
\end{proof}

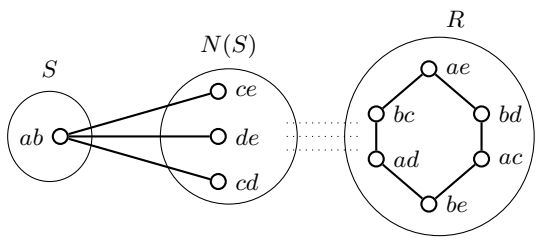
\begin{figure}
\centering
\begin{tikzpicture}
[lineDecorate/.style={-,thick},xscale=.7,yscale=.6,%
  nodeDecorate/.style={shape=circle,inner sep=2pt,draw,thick}]
\foreach \nodename/\x/\y/\direction/\navigate in {
  ab/0/1/left/west,
  ce/3/2/right/east,  de/3/1/right/east,  cd/3/0/right/east,
  ae/7/2.5/right/east, be/7/-.5/right/east,
  bc/6/1.5/right/east, ac/8/0.5/right/east,
   ad/6/0.5/right/east, bd/8/1.5/right/east}
{
  \node (\nodename) at (\x,\y) [nodeDecorate] {};
  \node [\direction] at (\nodename.\navigate) {\footnotesize$\nodename$};
}
\path
\foreach \startnode/\endnode in {ab/cd,ab/ce,ab/de,ae/bc,bc/ad,ad/be,be/ac,ac/bd,bd/ae}
{
  (\startnode) edge[lineDecorate] node {} (\endnode)
};
\draw (-.2,1) ellipse (0.8cm and 1cm);
\node at (-.2,2.5) {$S$};
\draw (3.2,1) ellipse (1.3cm and 1.5cm);
\node at (3.2,3) {$N(S)$};
\draw (7.2,1) ellipse (1.8cm and 2.2cm);
\node at (7.5,3.6) {$R$};
\draw[dotted] (4.3,1) -- (5.68,1);
\draw[dotted] (4.3,1.3) -- (5.68,1.3);
\draw[dotted] (4.3,0.7) -- (5.68,0.7);
\end{tikzpicture}
\caption{\label{fig:delta1}The case $\delta=1$.}
\end{figure}

  \begin{myclaim}\label{cl:delta1}
    All equitable partitions of $P$ with $\delta=1$ are orbit partitions.
  \end{myclaim}

  \begin{proof}
    Let $S$ be a singleton set in such an equitable partition. Similar
    to a previous argument, a cell cannot have vertices from both
    $N(S)$ and $V\backslash N(S)$. Therefore, any equitable partition
    refines the partition $S,N(S),R$ (see Figure \ref{fig:delta1}).
    Observe that $N(S)$ must be an independent set (otherwise there is
    a 3-cycle).  Moreover, if we assume that $S=\{ab\}$, $N(S)$ must
    be $\{ce,de,cd\}$ and therefore, $R=\{ae,be,ac,bc,ad,bd\}$ forms a
    6-cycle, as shown in the figure.  We proceed by further
    classifying equitable partitions on the basis of the partition
    induced by them inside $N(S)$.  Since $|N(S)|=3$, we have three
    possible cases. Either $N(S)$ is a cell, or it contains three
    1-cells, or it contains one singleton and one 2-cell.
%
%
  \begin{description}[leftmargin=1mm,labelsep=1.25mm,itemindent=0mm,topsep=1.5mm,itemsep=1.5mm]
  \item[\mdseries Case 1:] $N(S)$ is a cell.  We further classify the equitable
    partitions in this case on the basis of the partition induced on
    the set $R$.  First, we examine the possible cells $X$ in $R$
    which are compatible with $N(S)$. $X$ cannot be of size $1$ or
    $2$, otherwise $P[N(S),X]$ has at most two edges.  Also, $X$
    cannot be of size $4$ or $5$ since this would imply a cell of size
    $1$ or $2$ in $R$. Therefore, either $R$ is a cell, or there are
    two 3-cells in $R$.
    \begin{enumerate}%
[leftmargin=3mm,label=(\alph*),labelsep=1.25mm,itemindent=3mm,itemsep=1.5mm,topsep=1.5mm]
    \item $R$ is a cell.  The partition $\{ab\}$, $\{de,cd,ce\}$,
      $\{ac,ad,ae,bc,bd,be\}$ can be verified to be an equitable
      partition.  Moreover, it is easy to check that it is the orbit
      partition of the subgroup $S_{\{c,d,e\}} \times S_{\{a,b\}}$.
    \item The partition induced on $R$ is of the form $\{A,B\}$, where
      $|A|=|B|=3$. Because of regularity, the only possible $3$-cells
      in $R$ are the independent sets $\{ad,ac,ae\}$ and
      $\{bc,bd,be\}$.  The partition $\{ab\}$, $\{de,cd,ce\}$,
      $\{ad,ac,ae\}$, $\{bc,bd,be\}$ is clearly equitable. Moreover,
      it is easy to check that this partition is the orbit partition
      of the subgroup $S_{\{c,d,e\}}$.
    \end{enumerate}
  \item[\mdseries Case 2:] $N(S)$ contains three 1-cells.  Again, we classify
    the equitable partitions on the basis of the partition induced on
    the set $R$.  We can check that a cell of size more than two in
    $R$ will have at least one edge to some singleton in $N(S)$, and
    will be incompatible with that singleton.  Therefore, cells in $R$
    must have size at most $2$.  Moreover, any $2$-cell must be of the
    form $\{ax,bx\}$ for some $x\in\{d,c,e\}$ since all other 2-cells
    can be seen to be incompatible with some singleton cell in $N(S)$.
    Finally, it can be seen that every possible $1$-cell is
    incompatible with these three $2$-cells. Hence, $R$ must consist
    of three cells of size $2$, namely $\{ad,bd\}$, $\{ac,bc\}$,
    $\{ae,be\}$.  The partition $\{ab\}$, $\{cd\}$, $\{ce\}$,
    $\{de\}$, $\{ad,bd\}$, $\{ac,bc\}$, $\{ae,be\}$ can be easily seen
    to be equitable.  Moreover, it is easy to check that it is the
    orbit partition of the subgroup $S_{\{a,b\}}$.
  \item[\mdseries Case 3:] $N(S)$ contains a 2-cell $U=\{ce,de\}$ and a 1-cell
    $V=\{cd\}$.  Again, we need to classify the equitable partitions
    on the basis of the partition induced on the set $R$.  First, we
    examine the possible cells $X$ in $R$ which are compatible with
    $U$ and $V$. Clearly, $X$ cannot be a 5-cell since $P[X]$ cannot
    be regular.  It cannot be a 3-cell as well since the two candidate
    3-cells are the independent sets $\{ad,ac,ae\}$ and
    $\{bc,bd,be\}$.  Neither of them can be compatible with the
    singleton set $V$.  Also, $R$ cannot be a cell since it is
    incompatible with the singleton set $V$.  Moreover, the only
    possible $4$-cell is the neighborhood of the set $U$,
    i.e. $\{ac,bd,ad,bc\}$.  Any other $4$-cell is incompatible with
    $U$.  Overall, we have no cells of size 3, 5, or 6 in $R$.
    Therefore, we have only the following four remaining subcases.
    \begin{enumerate}%
[leftmargin=3mm,label=(\alph*),labelsep=1.25mm,itemindent=3mm,itemsep=1.5mm,topsep=1.5mm]
    \item $R$ consists of one 4-cell and two 1-cells.  This case is
      not possible since a $1$-cell cannot be compatible with a
      $4$-cell.
    \item $R$ consists of one 4-cell and one 2-cell. The cells are
      $\{ac,bd,ad,bc\}$ and $\{ae,be\}$.  The partition $\{ab\}$,
      $\{cd\}$, $\{ce,de\}$, $\{ae,be\}$, $\{ac,bd,ad,bc\}$ can be
      verified to be an equitable partition.  Moreover, it is easy to
      check that it is the orbit partition of the subgroup
      $S_{\{a,b\}} \times S_{\{c,d\}}$
    \item $R$ consists of three 2-cells. First, $ae$ and $be$ must be
      in the same $2$-cell, otherwise the cell containing any of them
      would be incompatible with $V$.  For the remaining vertices
      $ac,ad,bc,bd$, we can pair them up in three ways: (i) $ac,ad$
      and $bc,bd$, (ii) $ac,bc$ and $ad,bd$, or (iii) $ac,bd$ and
      $ad,bc$ The first case is not possible since $\{ae,be\}$ and
      $\{ac,ad\}$ are not compatible. The second case is not possible
      because $\{ac,bc\}$ and $U=\{ce,de\}$ are not compatible. The
      third case gives an equitable partition $\{ab\}$, $\{cd\}$,
      $\{ce,de\}$, $\{ae,be\}$, $\{ac,bd\}$, $\{ad,bc\}$.  Moreover,
      it is easy to check that it is the orbit partition of the
      subgroup generated by $(ab)(cd)$.
    \item $R$ consists of a bunch of $1$-cells and $2$-cells.
      Clearly, the vertices $ac,ad,$ $bc,bd$ cannot form a singleton
      cell, since such a 1-cell will not be compatible with $U$.
      Therefore, $\{ae\}$ and $\{be\}$ are the only possible singleton
      cells.  Neither of them can pair up with one of $ac,ad,bc,bd$
      since that cell would be incompatible with $V$. Therefore, they
      are forced to be singleton cells.  It remains to partition
      $ac,ad,bc,bd$ into two 2-cells.  The vertex $ac$ cannot be
      paired up with $bd$ or $bc$ since it will be incompatible with
      $be$. Therefore, the only possible case is to have 2-cells
      $\{ac,ad\}$ and $\{bc,bd\}$. The partition $\{ab\}$, $\{cd\}$,
      $\{ce,de\}$, $\{ae\}$, $\{be\}$, $\{ac,ad\}$, $\{bc,bd\}$ can be
      verified to be equitable. Moreover, it is easy to check that it
      is the orbit partition of the subgroup $S_{\{c,d\}}$.  (This
      case is identical to Case 2(b)).$~~\rule{1.0ex}{1.2ex}$\qed
    \end{enumerate}
  \end{description}
\end{proof}
\end{proof}

\subsection{The Johnson Graphs $J(n,2)$ are \tinhofer} \label{ssec:john-tin}

In this section, we show that the Johnson graphs $J(n,2)$ are
Tinhofer.  We begin with some necessary definitions.  Let $G$ be a
graph and denote the automorphism group of $G$ by $A$. For $v\in
V(G)$, by $A_v$ we denote the stabilizer subgroup of $A$ that fixes the vertex
$v$. Furthermore, for a subset $F\subset V(G)$, let $A_F=\bigcap_{v\in
  F}A_v$.  Let $\mathcal{P}_F$ denote the stable
partition of the colored version of $G$ where each vertex in $F$ is individualized.
Then the orbit partition of $A_F$ is a subpartition of
$\mathcal{P}_F$. Note that $G$ is Tinhofer if and only if, for every $F$, the orbit
partition of $A_F$ coincides with $\mathcal{P}_F$.

One way to prove that the two partitions coincide is to show that each
orbit $O$ of $A_F$ is definable in terms of $F$ in two-variable
first-order logic. Here, ``in terms of $F$'' means that a defining
formula $\Phi_O(x)$ can use constant symbols (names) for each vertex
in $F$. Furthermore, $\Phi_O(x)$ contains occurrences of only two variables, $x$
and $y$.  At least one occurrence of $x$ is free.  $\Phi_O(x)$ uses
two binary relation symbols $\sim$ and $=$ for adjacency and equality
of vertices. This formula is true on $G$ for $x=v$ exactly when $v\in O$.

Once $\Phi_O(x)$ is found for each $O$, the equality of the partitions
follows by a similar argument as in \cite[Theorem 1.8.1]{ImmermanL90}
or directly from the definitions of orbits, as those will imply that
any two orbits are separated by color refinement starting from the
individualization of $F$. The number of refinement steps sufficient to
separate $O$ from any other orbit can be only one greater than the
quantifier depth of~$\Phi_O(x)$.

In order to implement this scenario for $G=J(n,2)$, it will be
convenient to assume that $V(G)={[n]\choose2}$ (note, however, that the formulas
$\Phi_O(x)$ do not involve variables over $[n]$).  Given $\alpha\in
S_n$, by $\ell(\alpha)$ we denote the corresponding permutation of
$[n]\choose2$. Obviously, every $\ell(\alpha)$ is an automorphism of
$G$, and the automorphism group $A$ contains nothing else by the Whitney theorem~\cite{Whit32}.

Before designing the definitions $\Phi_O(x)$, we will need to make two
preliminary steps:
Describe $A_F$ and, then,
describe the orbits of $A_F$ (first irrespectively of any
  logical formalism; expressing these descriptions in two-variable
  first-order logic will be the next task).

We now proceed to the detailed proof. 

\begin{theorem}
$J(n,2)$ is a Tinhofer graph for all~$n$.  
\end{theorem}

\begin{proof}
  Note that $J(2,2)=K_1$, $J(3,2)=K_3$, and $J(4,2)$ is the octahedral
  graph, whose complement is $K(4,2)=3K_2$. Thus, these three graphs
  are amenable and, hence, Tinhofer. We can, therefore, assume that
  $n\ge5$.

  Call a fixed vertex $p\in F$ \emph{isolated} if $F$ contains no
  vertex adjacent to $p$. Let $F=F_1\cup F_2$ be the partition of $F$
  into non-isolated and isolated vertices. Furthermore, we define the
  partition
  \[
  [n]=W_1\cup W_2\cup W_3
  \]
  as follows: $W_1$ is the union of all non-isolated pairs $p$ (i.e.,
  all $p$ in $F_1$), and $W_2$ is the union of all isolated pairs $p$
  (i.e., all $p$ in $F_2$).
Thus, $W_3$ consists of the points of $[n]$ that are not included in any fixed pair.

  Note now that $\ell(\alpha)\in A_F$ if and only if $\alpha$ either fixes
or transposes the two points in each fixed pair.  It follows that $\ell(\alpha)\in A_F$
exactly when
  \begin{itemize}
  \item $\alpha(w)=w$ for every $w\in W_1$ and
  \item $\alpha(p)=p$ for every $p\in F_2$.
  \end{itemize}
  Given a vertex $u=\{a,b\}$ of $G$, let $O(u)$ denote its orbit with
  respect to $A_F$.  There are six kinds of orbits. Below we
  describe all of them along with providing suitable formal
  definitions $\Phi_{O(u)}(x)$.
  \begin{description}[leftmargin=1mm,labelsep=1.25mm,itemindent=0mm,topsep=1.5mm,itemsep=1.5mm]
  \item[\mdseries Case 1:] $\{a,b\}\subseteq W_1$.  Then $O(u)=\{u\}$.
    \emph{Formal definition:} $x=u$.
  \item[\mdseries Case 2:] $\{a,b\}\subseteq W_2$.  Here we have two subcases.  If
    $u\in F_2$, then $O(u)=\{u\}$ again.  Otherwise, $F_2$ contains
    two pairs $p_1=\{a,a'\}$ and $p_2=\{b,b'\}$.  In this subcase,
    \[
    O(u)=\{\{a,b\},\{a',b\},\{a,b'\},\{a',b'\}\},
    \]
    which is exactly the common neighborhood of $p_1$ and $p_2$.
    \emph{Formal definition:} $x\sim p_1\wedge x\sim p_2$.
  \item[\mdseries Case 3:] $\{a,b\}\subseteq W_3$.  Now $O(u)={W_3\choose2}$,
    which are exactly the non-fixed vertices with no neighbor in~$F$.
    \emph{Formal definition:} $\bigwedge_{p\in F}(x\ne p\wedge
    x\not\sim p)$.
  \item[\mdseries Case 4:] $a\in W_1$, $b\in W_2$. Let
    $p=\{b,b'\}$ be the pair in $F_2$ containing $b$. Then,
    \[
    O(u)=\{\{a,b\},\{a,b'\}\}.
    \]
To give a formal definition of $O(u)$, we consider two subcases.
    \begin{enumerate}%
[leftmargin=3mm,label=(\roman*),labelsep=1.25mm,itemindent=3mm,itemsep=1.5mm,topsep=1.5mm]
    \item
$a$ belongs to two adjacent vertices $q_1=\{a,a_1\}$ and
    $q_2=\{a,a_2\}$ in $F_1$.\\[.5mm]
    \emph{Formal definition:} $x\sim p\wedge x\sim q_1\wedge x\sim
    q_2$.  Indeed, the condition $x\sim p$ forces $x$ to contain
    either $b$ or $b'$.  This excludes the possibility that
    $x=\{a_1,a_2\}$ and, therefore, $x$ is forced to contain $a$ by
    the adjacency to $q_1$ and~$q_2$.
  \item 
$a$ belongs to a single vertex $q_1=\{a,a'\}$ in $F_1$.
By definition, $F_1$ contains also a vertex $q_2=\{a',a''\}$.
    \emph{Formal definition:} $x\sim p\wedge x\sim q_1\wedge x\not\sim q_2$.
    \end{enumerate}
  \item[\mdseries Case 5:] $a\in W_1$, $b\in W_3$.  Then
    $$O(u)=\setdef{\{a,b'\}}{b'\in W_3}.$$
Similarly to the preceding case, we distinguish two subcases.
   \begin{enumerate}%
[leftmargin=3mm,label=(\roman*),labelsep=1.25mm,itemindent=3mm,topsep=1.5mm,itemsep=1.5mm]
    \item
$a$ belongs to two adjacent vertices $q_1=\{a,a_1\}$ and
    $q_2=\{a,a_2\}$ in $F_1$.\\[.5mm]
\emph{Formal definition:}
 First of all, we say that $x\sim q_1\wedge x\sim q_2$.  It
    remains to exclude the possibility that $x\subseteq W_1\cup W_2$
    (in particular, this will exclude $x=\{a_1,a_2\}$ and force $x$ to
    contain $a$).  We do this by adding the following expression
    \begin{multline} {\!\!\bigwedge_{p\in F}x\ne p}\wedge\!\!\!
      {\bigwedge_{p,q\in F,p\not\sim q}\neg(x\sim p\wedge x\sim
        q)}\\\wedge{}  {\bigwedge_{p,q\in F_1,p\sim q}(x\sim
        p\wedge x\sim q\to\exists y\,(y\sim x\wedge y\sim p\wedge
        y\sim q))}.\label{eq:thirdpart}
    \end{multline}
    The first conjunctive term prevents $x$ to be one of the pairs in
    $F$.  The second term excludes the case that $x$ is covered by two
    disjoint pairs $p$ and $q$ in $F$. The third term excludes the
    case that $x$ is covered by two intersecting pairs $p$ and $q$ in
    $F$ or, equivalently, the case where $x$, $p$, and $q$ form a
    triangle.  It would be not enough just to forbid $x$, $p$, and $q$
    from forming a clique because this could also exclude a
    permissible case where $x$, $p$, and $q$ form a star (which is
    captured by the subformula beginning with $\exists y$).  Note,
    that we need the assumption $n\ge5$ in this place.
  \item 
$a$ belongs to a single vertex $q_1=\{a,a'\}$ in $F_1$, and
$q_2=\{a',a''\}$ is another vertex in~$F_1$.
    \emph{Formal definition:} $x\sim q_1\wedge x\not\sim q_2\wedge x\not\subseteq W_1\cup W_2$,
the last being expressed by the formula~\refeq{thirdpart}.
    \end{enumerate}
  \item[\mdseries Case 6:] $a\in W_2$, $b\in W_3$.  In this case, $F_2$ contains
    a pair $p=\{a,a'\}$ and
    \[
    O(u)=\setdef{\{a,b'\}}{b'\in W_3}\cup\setdef{\{a',b'\}}{b'\in
      W_3}.
    \]
    \emph{Formal definition:} $x\sim p\wedge x\not\subseteq W_1\cup
    W_2$, the latter being expressed by~\refeq{thirdpart}.
  \end{description}
  The proof is complete.\qed
\end{proof}

\ifproceedings{%
\subsection{A refinable non-Tinhofer graph}\label{ssec:tin-ref}

\proofSepRefTin
}{}%

\ifproceedings{%
\section{Proof of Theorem \protect\ref{thm:phard}}

\figphard
\phard*
\proofphard
}{}

\end{document}

Color refinement is a classical technique used to show that two given graphs G and H are non-isomorphic; it is very efficient, although it does not succeed on all graphs. We call a graph G amenable to color refinement if the color-refinement procedure succeeds in distinguishing G from any non-isomorphic graph H.  Tinhofer (1991) explored a linear programming approach to Graph Isomorphism and defined the notion of compact graphs: A graph is compact if its fractional automorphisms polytope is integral. Tinhofer noted that isomorphism testing for compact graphs can be done quite efficiently by linear programming. However, the problem of characterizing and recognizing compact graphs in polynomial time remains an open question.

Our results are summarized below:

– We determine the exact range of applicability of color refinement by showing that amenable graphs are recognizable in time O((n + m)logn), where n and m denote the number of vertices and the number of edges in the input graph.  – We show that all amenable graphs are compact. Thus, the applicability range for Tinhofer’s linear programming approach to isomorphism testing is at least as large as for the combinatorial approach based on color refinement.

– Exploring the relationship between color refinement and compactness further, we study related combinatorial and algebraic graph properties introduced by Tinhofer and Godsil. We show that the corresponding classes of graphs form a hierarchy and we prove that recognizing each of these graph classes is P-hard. In particular, this gives a first complexity bound for recognizing compact graphs.  and defined the notion of compact graphs: A graph is compact if its fractional automorphisms polytope is integral. Tinhofer noted that isomorphism testing for compact graphs can be done quite efficiently by linear programming. However, the problem of characterizing and recognizing compact graphs in polynomial time remains an open question. Our results are summarized below:

– We determine the exact range of applicability of color refinement by showing that amenable graphs are recognizable in time O((n + m)logn), where n and m denote the number of vertices and the number of edges in the input graph.

– We show that all amenable graphs are compact. Thus, the applicability range for Tinhofer’s linear programming approach to isomorphism testing is at least as large as for the combinatorial approach based on color refinement.

– Exploring the relationship between color refinement and compactness further, we study related combinatorial and algebraic graph properties introduced by Tinhofer and Godsil. We show that the corresponding classes of graphs form a hierarchy and we prove that recognizing each of these graph classes is P-hard. In particular, this gives a first complexity lower bound for recognizing compact graphs.

This version is identical to Version 2, except for the better formatted abstract and modified comments:

1. The incorrect Lemma 10 in the first version is now corrected (see Theorem 9 in the new version).

2. P-hardness proofs for the classes Discrete, Amenable, Compact, Tinhofer, and Refinable are included.

3. A graph separating the classes Tinhofer and Refinable is now included. We had left this open in the first version.